\documentclass[english,compsoc]{IEEEtran}
\usepackage[T1]{fontenc}
\usepackage{color}
\usepackage{float}
\usepackage{amsmath}
\usepackage{amsthm}
\usepackage{amssymb}
\usepackage{stackrel}
\usepackage{graphicx}

\usepackage{caption}
\makeatletter

\floatstyle{ruled}
\newfloat{algorithm}{tbp}{loa}
\providecommand{\algorithmname}{Algorithm}
\floatname{algorithm}{\protect\algorithmname}

\theoremstyle{plain}
\newtheorem{prop}{\protect\propositionname}

\usepackage{amsfonts}
\usepackage{cite}
\usepackage{hyperref}   
\hypersetup{
    colorlinks=true,    
    linkcolor=blue,       
    citecolor=blue,    
    urlcolor=blue,      
    filecolor=blue        
}

\makeatletter
\makeatletter
\def\eqref#1{{\color{blue}(\ref{#1})}}
\makeatother

\usepackage{array}
\usepackage{algorithm}
\usepackage{algorithmic}
\usepackage{subcaption}
 
\usepackage{stfloats}
\usepackage{graphicx}

\makeatother

\usepackage{babel}
\providecommand{\propositionname}{Proposition}

\begin{document}
\title{Co-Design of Sensing, Communications, and Control for Low-Altitude Wireless Networks}% with Finite Blocklength Transmission}
\author{ 
\IEEEauthorblockN{Haijia Jin, \IEEEmembership{Graduate Student Member, IEEE}},
\IEEEauthorblockN{Jun Wu, \IEEEmembership{Graduate Student Member, IEEE}},
\IEEEauthorblockN{Weijie Yuan, \IEEEmembership{Senior Member, IEEE}},
\IEEEauthorblockN{Fan Liu, \IEEEmembership{Senior Member, IEEE}}, \\ and \IEEEauthorblockN{Yuanhao Cui, \IEEEmembership{Member, IEEE}}

\thanks{
This work is supported in part by National Natural Science Foundation of China under Grant 62471208, in part by Guangdong Provincial Natural Science Foundation under Grant 2024A151510098, in part by Shenzhen Science and Technology Program under Grant JCYJ20240813094627037.
(\textit{Corresponding author: Weijie Yuan})

\noindent\hangindent=1.5em\hangafter=1%
         \noindent\textbullet~ H. Jin, J. Wu, W. Yuan, and Y. Cui are with the School of Automation and Intelligent  Manufacturing, Southern University of Science and Technology, Shenzhen 518055, China (e-mail: jinhj2024@mail.sustech.edu.cn, wuj2021@mail.sustech.edu.cn, yuanwj@sustech.edu.cn, cuiyh@sustech.edu.cn).%

\noindent\hangindent=1.5em\hangafter=1%  
         \noindent\textbullet~ F. Liu is with the National Mobile Communications Research Laboratory,
         Southeast University, Nanjing 210096, China (e-mail: f.liu@ieee.org).%
         
%\noindent\hangindent=1.5em\hangafter=1%  
 %        \noindent\textbullet~ Y. Cui is with the School of Information
 %and Communication Engineering, Beijing University of Posts and Telecom
%munications, Beijing 100876, China (e-mail: cuiyuan hao@bupt.edu.cn).

}
}

\IEEEtitleabstractindextext{
\begin{abstract}
The rapid advancement of Internet of Things (IoT) services and the evolution toward the sixth generation (6G) have positioned unmanned aerial vehicles (UAVs) as critical enablers of low-altitude wireless networks (LAWNs). This work investigates the co-design of integrated sensing, communication, and control ($\mathbf{SC^{2}}$) for multi-UAV cooperative systems with finite blocklength (FBL) transmission. In particular, the UAVs continuously monitor the state of the field robots and transmit their observations to the robot controller to ensure stable control while cooperating to localize an unknown sensing target (ST). To this end, a weighted optimization problem is first formulated by jointly considering the control and localization performance in terms of the linear quadratic regulator (LQR) cost and the determinant of the Fisher information matrix (FIM), respectively. The resultant problem,  optimizing resource allocations, the UAVs' deployment positions, and multi-user scheduling, is non-convex. To circumvent this challenge, we first derive a closed-form expression of the LQR cost with respect to other variables. Subsequently, the non-convex optimization problem is decomposed into a series of sub-problems by leveraging the alternating optimization (AO) approach, in which the difference of convex functions (DC) programming and projected gradient descent (PGD) method are employed to obtain an efficient near-optimal solution. Furthermore, the convergence and computational complexity of the proposed algorithm are thoroughly analyzed. Extensive simulation results are presented to validate the effectiveness of our proposed approach compared to the benchmark schemes and reveal the trade-off between control and sensing performance.
\end{abstract}

}
\maketitle
\IEEEdisplaynontitleabstractindextext

\section{Introduction}

\textcolor{black}{With the prosperity of Internet of Things (IoT) networks and the evolution toward the sixth generation (6G), low-altitude wireless networks (LAWNs) have emerged as a critical enabler for future wireless communications \cite{LA-UAV1}. Low-altitude unmanned aerial vehicles (UAVs) and electric vertical take-off and landing (eVTOL) aircraft, as core components of LAWN, have attracted significant research interest, which are uniquely suited to address the connectivity challenges in dense urban environments, fluctuating network demands, and real-time communication requirements \cite{UAV_reivew2}. By operating closer to ground users (GUs), they provide enhanced line-of-sight (LoS) communication, reduced latency, and improved adaptability to dynamic environments \cite{LAWN_UAV1}. Their mobility and rapid deployment capabilities further enable flexible network coverage, making LAWN a key technology for supporting large-scale IoT applications and next-generation wireless networks.}

%Before revision: With the prosperity of Internet of Things (IoT) networks and the evolution toward the sixth generation (6G), low-altitude unmanned aerial vehicles (UAVs) and electric vertical take-off and landing (eVTOL) aircraft have garnered considerable research interest \cite{UAV_reivew2}. Operating as a core component of low-altitude wireless networks (LAWN), these UAVs are uniquely positioned to address the connectivity challenges posed by dense urban environments, variable network demands, and real-time communication needs, which operate closer to ground users (GUs), offering enhanced line-of-sight (LoS) links, reduced signal latency, and greater adaptability to dynamic environments \cite{LA-UAV1}. These advantages, combined with their mobility and rapid deployment capabilities, enable them to provide flexible network coverage and support large-scale IoT applications within future LAWN.

Building on these strengths, low-altitude UAVs offer significant advantages in wireless communication by serving as airborne communication nodes \cite{UAV_Cui}. Existing research has primarily concentrated on optimizing UAV deployment, resource management, and user association to achieve superior communication performance. For example, the deployment of multi-UAVs has been investigated in \cite{UAV_communication} for IoT networks by jointly optimizing the mobility of the UAVs, device-UAV association, and power control. Furthermore, Zhou \emph{et al.} \cite{UAV_com_moving1} and Cai \emph{et al.} \cite{UAV_com_moving2} proposed adaptive trajectory design strategies for moving UAVs to enhance communication efficiency and security. In addition to enabling wireless communication services, sensor-equipped UAVs have been extensively utilized for various sensing tasks, including target localization, tracking, and navigation \cite{Sensing_Cui}. For instance, target localization in single-UAV sensing networks has been analyzed by incorporating both angle-of-arrival (AOA) and position errors \cite{UAV_sensing1}. Moreover, Xu \emph{et al.} \cite{UAV_sensing2} extended the path optimization for AOA-based target localization to multi-UAV scenarios. Wu \emph{et al.} \cite{UAV_sensing3} further proposed a vision-based framework to enable real-time aerial target localization and tracking within UAV sensing systems. However, the majority of existing studies consider communication and sensing as separate processes, resulting in inefficient resource utilization and restricted overall performance \cite{UAV_S&C5}. 

To tackle this challenge, the concept of integrated sensing and communication (ISAC) has emerged as a distributive technology to LAWNs, allowing UAVs to seamlessly perform communication and sensing tasks using shared wireless resources \cite{ISAC_Cui}.  Compared to dedicated sensing and communication systems, ISAC systems enhance spectral efficiency while simultaneously reducing hardware costs \cite{ISAC_Gong}. To fully harness the potential of ISAC systems, numerous studies have delved into evaluating the performance of both sensing and communication (S\&C). On the one hand, communication performance is typically measured by metrics such as communication rates, channel capacity, and bit error rate (BER) \cite{UAV_com_moving1}, \cite{UAV_com_moving2}. On the other hand, sensing performance is often quantified using estimation error, characterized by the mean squared error (MSE), which is bounded by the Cramér–Rao bound (CRB) \cite{ISAC_Cui2}. Specifically, the CRB-rate region has been introduced as a fundamental framework to characterize the trade-off between sensing and communication \cite{ISAC1}. To jointly consider the S\&C performance, beamforming design in ISAC systems has been studied, aiming to minimize the CRB for target localization while guaranteeing the Quality of Service (QoS) for communication \cite{ISAC2}. Furthermore, recent research has investigated the application of ISAC in UAV systems, focusing on optimizing UAV positions and trajectories to achieve the S\&C trade-off \cite{UAV_S&C1,UAV_S&C2}. For example, the trade-off between communication and sensing was analyzed through a weighted trajectory optimization for single-UAV systems, where the CRB was employed as the sensing performance metric \cite{UAV_S&C3}. Wu \emph{et al.} \cite{UAV_S&C4} further proposed dynamic real-time trajectories for multi-UAV systems to enable adaptive communication and sensing capabilities.

Beyond their role in integrating sensing and communication, UAVs capitalize on their inherent mobility and adaptability to play a pivotal role in advanced control applications, meeting critical challenges such as mastering six degrees of freedom (DoF), overcoming the algorithmic complexity of UAV swarm coordination, and ensuring ultra-reliable low-latency communication (URLLC) for controlling high-mobility targets \cite{UAV_control_Cui2}, \cite{UAV_control_Cui1}.
% Beyond their role in integrating sensing and communication, UAVs capitalize on their inherent mobility and adaptability to play a pivotal role in advanced control applications, encompassing areas such as efficient computer vision for real-time navigation, networked computing strategies for distributed control, and traditional aircraft-related challenges like
% collision avoidance and formation flight \cite{UAV_cont1}.
For instance, the robustness and tracking performance of quadrotor stabilization was significantly improved through an innovative approach to adjust the weighting matrices in linear quadratic regulator (LQR) and LQR-proportional integral (PI) controllers \cite{UAV_cont3}. In addition, various control strategies have been applied to address a wide range of challenges specific to rotorcraft or rotary-wing UAV systems \cite{UAV_cont4}. Although significant progress has been achieved in control strategies, investigating wireless control systems remains equally essential, which must account for not only precise control strategies but also advanced communication capabilities to guarantee reliable performance. Nair \emph{et al.} \cite{cont&comm1} explored the trade-off between communication and control, highlighting that the data throughput must exceed the intrinsic entropy rate to stabilize the noisy linear control system. Moreover, Kostina \emph{et al.} \cite{rate_cost_trade_off} derived a lower bound on the communication rate corresponding to a target LQR cost, which serves as a quantitative measure of control system performance. Taking into account the relationship between the data throughput and the LQR cost, a control-oriented optimization problem was investigated in a sensing-communication-computing-control ($\mathbf{SC^{3}}$) integrated satellite-UAV system \cite{cont&comm4}.

In addition, the ultra-stringent latency and highly reliable requirements of real-time control make it imperative to explore novel communication technologies. Short-packet communication with the finite blocklength (FBL) transmission has shown significant potential in addressing the requirements of latency-sensitive services in IoT networks \cite{FBL}, \cite{FBL11}. The importance of FBL transmission stems from its ability to meet the stringent low-latency requirements of UAV-enabled systems, where operational periods are typically segmented into multiple short time slots to optimize UAV trajectory design \cite{UAV_FBL1}.
\textcolor{black}{Compared to infinite blocklength (IBL) transmission, where the achievable rate follows Shannon capacity, FBL transmission introduces a more complex rate function that depends on signal-to-interference-plus-noise ratio (SINR), block error rate (BLER), and blocklength \cite{FBL_rATE}. This added complexity poses significant challenges for optimization, as the rate function in FBL scenarios is non-convex and highly sensitive to system parameters, making resource allocation more difficult than in IBL systems \cite{general_solulation}.}
FBL transmission has been widely studied in UAV-enabled systems to support latency-sensitive applications. In \cite{UAV_FBL3} and \cite{UAV_FBL2}, the communication secrecy and covertness were investigated for a UAV-enabled wireless system with FBL transmission. Raut \emph{et al.} \cite{UAV_FBL4} focused on a UAV-assisted nonlinear energy harvesting full-duplex network with both infinite and finite blocklength transmission, aiming to optimize reliability in IoT scenarios.

\textcolor{black}{However, existing works on UAV-enabled communication with FBL transmission, sensing, and control ignore the integration of these functionalities, especially in multi-UAV cooperative systems where co-design of sensing, communication, and control ($\mathbf{SC^{2}}$) becomes essential. As UAVs are required to simultaneously deliver downlink signals to GUs and perform sensing tasks like target localization, developing a co-design $\mathbf{SC^{2}}$ framework under FBL transmission remains a significant yet underexplored challenge. Motivated by this gap, this paper investigates a three-dimensional (3D) multi-UAV cooperative $\mathbf{SC^{2}}$ network under FBL transmission. Specifically, a 3D scenario is considered, where UAVs transmit the observation results to the field robots via ISAC signals while simultaneously performing sensing tasks by receiving reflected echoes from a sensing target (ST). For simplicity, it is assumed that reflections unrelated to the ST and the echo interference between UAVs are negligible, as the distance between multiple UAVs is relatively small compared to the distance between UAVs and the ST, leading to a minimal impact on the overall sensing performance.}
% Before revision: However, existing works on UAV-enabled communication with FBL transmission, sensing, and control ignore the integration of these functionalities, especially in multi-UAV cooperative systems where co-design of sensing, communication, and control ($\mathbf{SC^{2}}$) becomes essential. As UAVs are required to simultaneously deliver downlink signals to GUs and perform sensing tasks like target localization, developing a co-design $\mathbf{SC^{2}}$ framework under FBL transmission remains a significant yet underexplored challenge. Motivated by this gap, this paper investigates a three-dimensional (3D) multi-UAV cooperative $\mathbf{SC^{2}}$ network under FBL transmission. Specifically, a 3D scenario is considered, where UAVs transmit the observation results to the field robots via ISAC signals while simultaneously performing sensing tasks by receiving reflected echoes from a sensing target (ST).  For simplicity, it is assumed that reflections unrelated to the sensing target (ST) and echo interference between UAVs are negligible. 
The main contributions of this work are summarized as follows.
\begin{itemize}
\item We consider the determinant of the Fisher information matrix (FIM) as the performance metric for the sensing system and derive the expression of the FIM for the 3D coordinates of the ST. In addition, the infinite horizon LQR cost is adopted as the control performance metric, and a relationship is established between the LQR cost and the communication rate under FBL transmission in a partially observed control system.
\item By integrating these metrics, this work formulates a weighted optimization problem for the multi-UAV cooperative $\mathbf{SC^{2}}$  FBL network, aiming to minimize the LQR cost and maximize the determinant of the FIM for the ST. This is accomplished by jointly optimizing UAV-robot association, power allocation, and UAV positions subject to constraints, including rate-LQR cost bounds, power budgets, collision avoidance, and flight boundaries. Moreover, a weighting factor is introduced to balance the trade-off between control and sensing performance. By adjusting the weighting factor, UAV positions design and power allocation can dynamically prioritize between control and sensing performance.
\item The optimization problem is inherently non-convex, rendering it unsolvable through direct methods. To address this challenge, we first derive the closed-form of LQR cost with respect to UAV-robot association, power allocation, and UAV positions to reduce the dimensionality of the problem and then separate it into a series of sub-problems. Leveraging the alternative optimization (AO) method, we propose an effective algorithm to solve the formulated problem by jointly employing difference-of-convex (DC) programming and projected gradient descent (PGD) techniques. The convergence and computational complexity of the proposed algorithm are further analyzed to demonstrate its validity and effectiveness.
\item Finally, simulation results show that the proposed algorithm can adaptively adjust UAV positions and power allocation strategies according to the rate-LQR cost bounds, power budget, and weighting factor. Furthermore, the proposed algorithm can achieve enhanced sensing and control performance compared to the benchmark schemes.
\end{itemize}
\begin{figure*}
\includegraphics[width=1\linewidth]{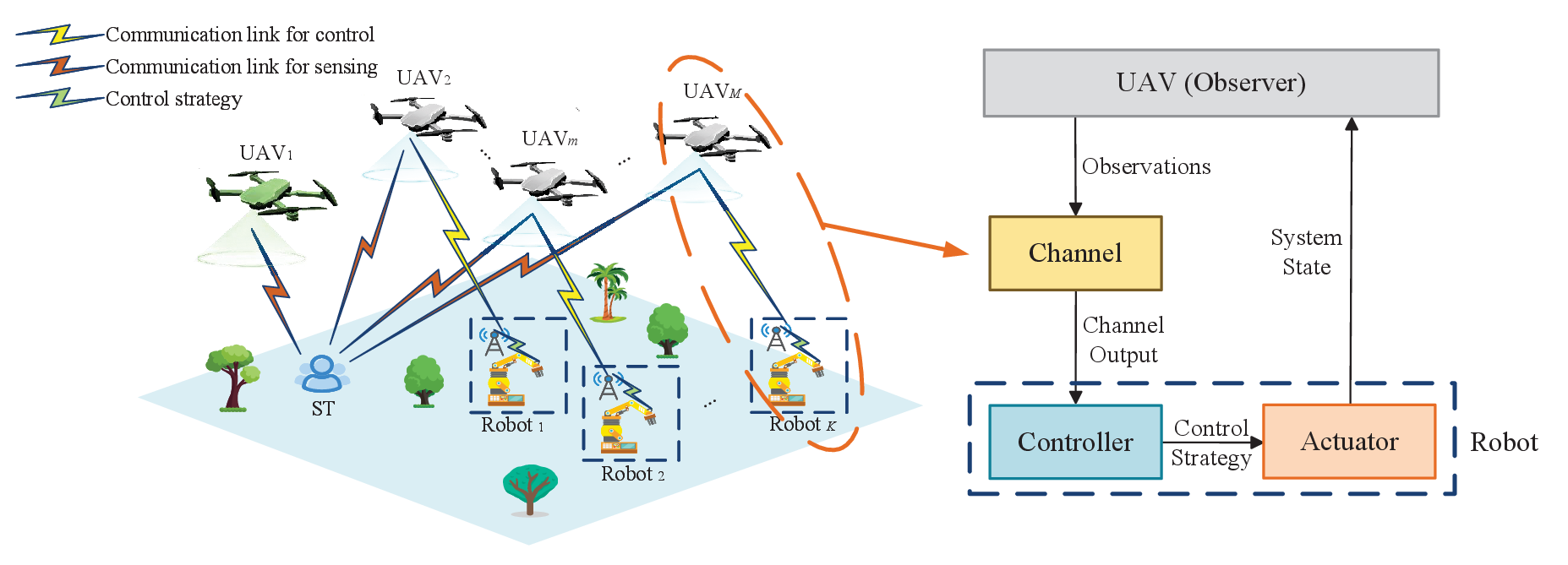}
\captionsetup{justification=raggedright, singlelinecheck=false}
\caption{The considered multi-UAV cooperative $\mathbf{SC^{2}}$ network.}
\label{fig:sys_mod}
\end{figure*}

% \begin{figure}
% \includegraphics[width=0.95\linewidth]{a_plot_reference/UAV_enabled_C_L}
% \caption{The considered Multi-UAV-enabled $\mathbf{SC^{2}}$ network.}
% \label{fig:sys_mod}
% \end{figure}

The rest of this article is structured as follows. Section \ref{sec:2} introduces the multi-UAV cooperative $\mathbf{SC^{2}}$ network with FBL transmission. Section \ref{sec:Problem-formulation} formulates the joint optimization of the LQR cost, power allocation, UAV-robot association, and UAV positions problem, which is subsequently addressed in Section \ref{sec:Proposed-solutions}. Simulation results are provided in Section \ref{sec:4}, followed by the conclusions in Section \ref{sec:5}.

\textit{Notations: } Unless otherwise specified, bold lowercase and uppercase letters, e.g., $\mathbf{a}$ and $\mathbf{A}$, denote vectors and matrices, respectively. $\mathbb{R}^{M}$ indicates the $M$-dimensional real-valued column vectors. $\left(\cdot\right)^{T}$ and $\left(\cdot\right)^{-1}$ denote the transposition and inverse operation, respectively. $\|\cdot\|$, $\textrm{tr}\left(\cdot\right)$, $\textrm{det}\left(\cdot\right)$, and $\mathbb{E}\{\cdot\}$ are respectively the $\ell_{2}$-norm, trace, determinant, and expectation of a matrix. $\mathcal{N}\left(a,b\right)$ is used to represent a Gaussian distribution of mean $a$ and variance $b$. $\textrm{Diag}\left(\cdot\right)$ denotes a diagonal matrix, and $\mathbf{I}_{M}$ is the $M\times M$ identity matrix. The function $Q^{-1}(\cdot)$ is the inverse of Gaussian $\mathcal{Q}$ function, i.e., $Q(x)=\frac{1}{\sqrt{2\pi}}\int_{x}^{\infty}\exp(-\frac{t^{2}}{2})\textrm{d}t$.

\section{System Model \label{sec:2}}
As depicted in Fig. \ref{fig:sys_mod}, we focus on a multi-UAV cooperative $\mathbf{SC^{2}}$ network in which $M$ single-antenna UAVs monitor $K$ robots over the same frequency band channel while simultaneously localizing a single ST. The UAVs and robots are indexed by $\ensuremath{\mathcal{M}=\left\{ 1,2,\ldots,M\right\} }$ and $\ensuremath{\mathcal{K}=\left\{ 1,2,\ldots,K\right\} }$, respectively. Specifically, the UAVs send the observations to robots via ISAC signals, based on which the robots use these signals to generate optimal control strategies for their actuators, e.g., robotic arms. At the same time, the transmitted ISAC signals can also be radiated to the ST for sensing purposes.\footnote{\textcolor{black}{To better investigate the trade-off between sensing and control performance, we assume that ST does not engage in direct communication with UAVs. Future research will extend this framework by considering the ST as an active communication participant, facilitating a closed-loop $\mathbf{SC^{2}}$ system.}} Without loss of generality, we consider a 3D system, where the coordinates of the $m$-th UAV, $m\in\mathcal{M}$, is denoted as $\mathbf{q}_{m}=\left[x_{m}^{q},y_{m}^{q},z_{m}^{q}\right]^{T}\in\mathbb{R}^{3}$. The locations of the $k$-th robot, $k\in\mathcal{K}$, and the ST are represented by  $\mathbf{u}_{k}=\left[x_{k}^{u},y_{k}^{u},z_{k}^{u}\right]^{T}\in\mathbb{R}^{3}$ and $\mathbf{s}=\left[x^{s},y^{s},z^{s}\right]^{T}\in\mathbb{R}^{3}$, respectively. Furthermore, it is assumed that the channels from UAVs to the robots are dominated by the LoS link, as commonly adopted in \cite{UAV_S&C3}, \cite{UAV_S&C4}, and \cite{LoS}.

\subsection{Control Model\label{subsec:Control-Model}}

As for the control part, we model each robot  as a linear control system and formulate the discrete-time system equation of the $k$-th control system as
\begin{equation}
\mathbf{x}_{k,n+1}=\mathbf{A}_{k}\mathbf{x}_{k,n}+\mathbf{B}_{k}\mathbf{z}_{k,n}+\mathbf{v}_{k},\label{eq:system_equation}
\end{equation}
where $\mathbf{x}_{k,n}\in\mathbb{R}^{\iota}$ denotes the system state, $\mathbf{z}_{k,n}\in\mathbb{R}^{\kappa}$ denotes the control action, $\mathbf{v}_{k}\in\mathbb{R}^{\iota}$ is the control system noise with zero mean and covariance matrix $\boldsymbol{\Sigma}_{k}^{v}$, and $\mathbf{A}_{k}\in\mathbb{R}^{\iota\times\iota}$ and $\mathbf{B}_{k}\in\mathbb{R}^{\iota\times\kappa}$ are fixed state evolution matrix and input matrix, respectively, with $\iota$ and $\kappa$ being the dimensions of the system state and control action input, respectively. At time instant $n$, the $k$-th robot controller receives the observation results and generates the optimal control action $\mathbf{z}_{k,n}$ for the actuator based on the received data up to time $n$. Moreover, serving as observers, UAVs monitor the system state, such that we have the following observation equation
\begin{equation}
\mathbf{y}_{k,n}=\mathbf{C}_{k}\mathbf{x}_{k,n}+\mathbf{w}_{k},\label{eq:observation_equation}
\end{equation}
where $\mathbf{C}_{k}\in\mathbb{R}^{\zeta\times\iota}$ is a deterministic matrix with $\zeta$ representing the dimension of the observation results and $\mathbf{w}_{k}\in\mathbb{R}^{\zeta}$ denotes the observation noise with zero mean and covariance matrix $\boldsymbol{\Sigma}_{k}^{w}$. The LQR cost function is generally employed to access the control performance. 
\textcolor{black}{In particular, we consider the infinite horizon LQR cost to evaluate control performance, which is given by \cite{rate_cost_trade_off}
\begin{equation}
b_{k}\triangleq\sup\underset{N\to\infty}{\lim}\mathbb{E}\left\{ \left[\frac{1}{N}\sum\limits _{n=1}^{N}\Bigl(\mathbf{x}_{k,n}^{\mathrm{T}}\mathbf{Q}_{k}\mathbf{x}_{k,n}+\mathbf{z}_{k,n}^{\mathrm{T}}\mathbf{R}_{k}\mathbf{z}_{k,n}\Bigr)\right]\right\} ,\label{eq:LQR_cost}
\end{equation} 
where $N$ represents the time horizon over which the expected cost is averaged, and $\mathbf{Q}_{k}\succeq0$ and $\mathbf{R}_{k}\succeq0$ are weight matrices that balance the trade-off between the deviation of the system from the optimal desired state zero and the power consumption. The LQR cost function jointly takes into account system state convergence and control energy consumption. A lower $b_{k}$ indicates better system stability and less control energy consumption, leading to improved control performance. }

\subsection{Communication Model\label{subsec:Communication-Model}}
After monitoring the state of the control system, UAVs transmit the observation results to the field robots via ISAC signals. Let us define the distance between the $m$-th UAV and the $k$-th robot as
\begin{equation}
\begin{split}d_{m,k} & =\parallel\mathbf{q}_{m}-\mathbf{u}_{k}\parallel\\
 & =\sqrt{\left(x_{m}^{q}-x_{k}^{u}\right)^{2}+\left(y_{m}^{q}-y_{k}^{u}\right)^{2}+\left(z_{m}^{q}-z_{k}^{u}\right)^{2}}
\end{split}
.
\end{equation}
Considering the LoS channel, the free-space path loss model is applicable to the communication model \cite{LoS}. Thus, the channel gain between the $m$-th UAV and the
$k$-th robot can be written as
\begin{equation}
h_{m,k}=\frac{\alpha_{0}}{\left(d_{m,k}\right)^{2}},
\end{equation}
where $\alpha_{0}=\frac{G_{\textrm{T}}G_{\textrm{c}}\lambda^{2}}{\left(4\pi\right)^{2}}$ denotes the channel power at the reference distance $d_{m,k}=1\ \textrm{m}$ with $G_{\textrm{T}}$ representing the UAV transmitting antenna gain, $G_{\textrm{c}}$ the robot receiving antenna gain, and $\lambda$ the wavelength. Furthermore, we assume that each UAV is only allowed to serve at most one robot, and each robot is served by exactly one UAV. For ease of exposition, we introduce a binary matrix $\boldsymbol{\Theta}$ to characterize the associations from the UAVs to the robots. Specifically, the $\left(m,k\right)$-th entry of $\boldsymbol{\Theta}$, denoted by $\theta_{m,k}=\{0,1\}$, indicates the association status between the $m$-th UAV and the $k$-th robot, i.e., $\theta_{m,k}={1}$ represents that the $m$-th UAV serves the $k$-th robot and vice versa not. Thus, if the $k$-th robot is served by the $m$-th UAV, the SINR at the $k$-th robot can be formulated as

\begin{equation}
\Gamma_{m,k}=\frac{p_{m}h_{m,k}}{\underset{i\in\mathcal{M},i\neq m}{\sum}p_{i}h_{i,k}+\sigma_{k}^{2}},\label{SINR}
\end{equation}
where $p_{m}$ denotes the power allocated to each UAV in the network,  $\sum_{i\in\mathcal{M},i\neq m}p_{i}h_{i,k}$ is the co-channel interference, and $\sigma_{k}^{2}$ represents the noise power at the $k$-th robot. Thus, the achievable rate can be expressed as $R_{m,k}=\log\left(1+\Gamma_{m,k}\right)$. It is worth noting that the transmission delay is capable of significantly impacting the control performance. Consequently, we advocate exploiting the FBL transmission to fulfill the stringent delay requirements in this work. The communication rate for FBL transmission depends on the SINR, the BLER, and the blocklength, such that the data throughput between the $m$-th UAV and the $k$-th robot can be rewritten as \cite{FBL_rATE} 
\begin{equation}
R_{m,k}=\log\left(1+\Gamma_{m,k}\right)-\sqrt{\frac{V_{m,k}}{l_{m,k}}}Q^{-1}(\ensuremath{\epsilon}),
\end{equation}
where $l_{m,k}$ denotes the blocklength of the transmitted signal, $\epsilon$ is the BLER, and $V_{m,k}$ represents the channel dispersion, which is defined as
\begin{equation}
V_{m,k}=1-\left(1+\Gamma_{m,k}\right)^{-2}.
\end{equation}

\subsection{Sensing Model}

Let us denote the  distance between the $m$-th UAV and the ST as
\begin{equation}
\begin{split}d_{m} & =\parallel\mathbf{q}_{m}-\mathbf{s}\parallel\\
 & =\sqrt{\left(x_{m}^{q}-x^{s}\right)^{2}+\left(y_{m}^{q}-y^{s}\right)^{2}+\left(z_{m}^{q}-z^{s}\right)^{2}}
\end{split}
,
\end{equation}
which can be measured by $\frac{\tau_{m}c}{2}$, where $c$ is the speed of light and $\tau_{m}$ represents the round-trip delay of the ISAC signal arriving from the ST. Without loss of generality, we consider that the measurements are impacted by Gaussian noises, such that the measurement of $d_{m}$ is given by 
\begin{equation}
\hat{d}_{m}=d_{m}+\upsilon_{m},\label{eq:d_measure}
\end{equation}
where $\upsilon_{m}\sim\mathcal{N}\left(0,\sigma_{m}^{2}\right)$ is the additive white Gaussian noise (AWGN) with zero mean and variance $\sigma_{m}^{2}$. In general, the variance $\sigma_{m}^{2}$ is inversely proportional to the signal-to-noise ratio (SNR) at the UAV, i.e.,  \cite{sensing_noise}
\begin{equation}
\Gamma_{m}=\frac{p_{m}G_{{\textrm{p}}}h_{m}}{\sigma_{0}^{2}},\label{SNR at the UAV}
\end{equation}
where 
% We use $\mathbf{p}=\left[p_{1},p_{2},\ldots,p_{M}\right]^{T}$  to denote
% the power allocation vector.
$G_{{\textrm{p}}}$ is the signal processing gain at the UAV, and $\sigma_{0}^{2}$ represents the noise power at the receiver. The two-way channel power gain between the UAV and the ST is given by $h_{m}\triangleq\frac{\beta_{0}}{\left(d_{m}\right)^{4}}$ with $\beta_{0}\triangleq\frac{G_{\textrm{T}}G_{\textrm{r}}\sigma_{\textrm{rcs}}\lambda^{2}}{\left(4\pi\right)^{3}}$ representing the channel power at the reference distance $d_{m}=1\ \textrm{m}$. Here, $G_{\textrm{r}}$ denotes the ST receiving antenna gain, and $\sigma_{\textrm{rcs}}$ is the Radar Cross-Section (RCS). Consequently, $\sigma_{m}^{2}$ can be specified as
\begin{equation}
\sigma_{m}^{2}=\frac{\rho\sigma_{0}^{2}\left(d_{m}\right)^{4}}{p_{m}G_{p}\beta_{0}},
\end{equation}
where $\rho$ is a constant related to the system settings \cite{UAV_S&C3}. 

% \begin{figure}
% \centering
% \includegraphics[width=0.7\linewidth]{a_plot_reference/control}
% \caption{\label{fig:sys_mod-control}Discrete-time stochastic linear system
% over a communication channel}
% \end{figure}

\section{Problem   Formulation \label{sec:Problem-formulation}}

This work aims to minimize the sum LQR cost of the control system while maximizing the localization accuracy of the ST. Regarding control performance, \eqref{eq:LQR_cost} shows that the LQR cost is dictated by the system state $\mathbf{x}_{k,n}$ and the control input $\mathbf{z}_{k,n}$, both of which are primarily determined by the control strategy derived from the observations received by the robots. However, extensive research has already focused on detailed control strategies. Therefore, this work shifts attention to how communication capabilities influence the LQR cost in such a wireless control scenario. According to \cite{rate_cost_trade_off}, achieving a desired LQR cost $b_{k}$ requires the received data throughput at the $k$-th robot to satisfy the rate-LQR cost bound

\begin{equation}
B\underset{m\in\mathcal{M}}{\sum}\theta_{m,k}R_{m,k}\geq L_{k},\forall k\in\mathcal{K},\label{eq:rate_cost_trade-off}
\end{equation}
where $B$ represents the channel bandwidth, and $L_{k}$ is defined as 
\begin{equation}
L_{k}\triangleq g_{k}+\frac{\iota}{2}\log\left(1+\frac{\iota\left(\det\mathbf{N}_{k}\mathbf{M}_{k}\right)^{\frac{1}{\iota}}}{b_{k}-\left(b_{k}\right)_{\min}}\right),\forall k\in\mathcal{K},\label{eq:rate_cost_trade-off-1}
\end{equation}
where $g_{k}\triangleq\log\mid\det\mathbf{A}_{k}\mid$ denotes the intrinsic entropy rate, indicating the stability of the $k$-th robot control system. Meanwhile, $\left(b_{k}\right)_{\min}$ represents the minimum LQR cost attainable in the absence of communication constraints and is defined as
\begin{equation}
\left(b_{k}\right)_{\min}=\textrm{tr}\left(\boldsymbol{\Sigma}_{k}^{v}\mathbf{S}_{k}\right)+\textrm{tr}\left(\boldsymbol{\Sigma}_{k}\mathbf{S}_{k}\mathbf{A}_{k}^{T}\mathbf{M}_{k}\mathbf{A}_{k}\right).\label{bmin}
\end{equation}
The terms $\mathbf{S}_{k}$, $\mathbf{N}_{k}$, $\mathbf{M}_{k}$, and $\boldsymbol{\Sigma}_{k}$ are solutions to some algebraic Riccati equations concerning the given control parameters, i.e., $\mathbf{A}_{k}$, $\mathbf{B}_{k}$, $\mathbf{R}_{k}$, $\mathbf{Q}_{k}$, $\boldsymbol{\Sigma}_{k}^{v}$, and $\boldsymbol{\Sigma}_{k}^{w}$ \cite{rate_cost_trade_off}. Specifically, 
$\mathbf{S}_{k}$ and $\mathbf{M}_{k}$ are solutions to the following algebraic Riccati equations, i.e.,
\begin{equation}
\mathbf{S}_{k}=\mathbf{Q}_{k}+\mathbf{A}_{k}^{T}\left(\mathbf{S}_{k}-\mathbf{M}_{k}\right)\mathbf{A}_{k},
\end{equation}
\begin{equation}
\mathbf{M}_{k}=\mathbf{S}_{k}\mathbf{B}_{k}\left(\mathbf{R}_{k}+\mathbf{B}_{k}^{T}\mathbf{S}_{k}\mathbf{B}_{k}\right)^{-1}\mathbf{B}_{k}^{T}\mathbf{S}_{k},
\end{equation}
and $\boldsymbol{\Sigma}_{k}$ is derived using the Kalman filter and denoted as
\begin{equation}
\boldsymbol{\Sigma}_{k}=\mathbf{P}_{k}-\mathbf{K}_{k}\left(\mathbf{C}_{k}\mathbf{P}_{k}\mathbf{C}_{k}^{T}+\boldsymbol{\Sigma}_{k}^{w}\right)\mathbf{K}_{k}^{T},
\end{equation}
where $\mathbf{P}_{k}$ is the solution to the algebraic Riccati equation
\begin{equation}
\mathbf{P}_{k}=\mathbf{A}_{k}\mathbf{P}_{k}\mathbf{A}_{k}^{T}-\mathbf{A}_{k}\mathbf{K}_{k}\left(\mathbf{C}_{k}\mathbf{P}_{k}\mathbf{C}_{k}^{T}+\boldsymbol{\Sigma}_{k}^{w}\right)\mathbf{K}_{k}^{T}\mathbf{A}_{k}^{T}+\boldsymbol{\Sigma}_{k}^{v},
\end{equation}
and $\mathbf{K}_{k}$ denotes the Kalman filter gain, defined as
\begin{equation}
\mathbf{K}_{k}=\mathbf{P}_{k}\mathbf{C}_{k}^{T}\left(\mathbf{C}_{k}\mathbf{P}_{k}\mathbf{C}_{k}^{T}+\boldsymbol{\Sigma}_{k}^{w}\right)^{-1}.
\end{equation}
The partially observed steady-state covariance matrix $\mathbf{N}_{k}$ is given by 
\begin{equation}
\mathbf{N}_{k}=\mathbf{A}_{k}\boldsymbol{\Sigma}_{k}\mathbf{A}_{k}^{T}-\boldsymbol{\Sigma}_{k}+\boldsymbol{\Sigma}_{k}^{v}.
\end{equation}
These algebraic Riccati equations are solved iteratively using numerical methods, where the equations are updated starting from an initial estimate until convergence within a predefined tolerance is achieved. Moreover, note that the right-hand-side (RHS) of \eqref{eq:rate_cost_trade-off-1} decreases with respect to $b_{k}$. Therefore, a smaller $b_{k}$ requires a higher communication rate between the UAV and the robot, which corresponds to improved control performance.

As for localization accuracy of the ST, MSE, defined as $\mathbb{E}\left(\parallel\mathbf{s}-\hat{\mathbf{s}}\parallel^{2}\right)$, is typically adopted as the metric. However, obtaining a closed-form expression for MSE and minimizing it is both challenging and computationally complex. Instead, we utilize the determinant of the FIM of $\mathbf{s}$, which is denoted as $\det\boldsymbol{\Phi}_{\mathbf{s}}$, to assess the performance of the localization system \cite{det_FIM}. To compute $\boldsymbol{\Phi}_{\mathbf{s}}$, the FIM for the distances between the UAVs and the ST is derived first, which is denoted as  $\boldsymbol{\Phi}\left(\mathbf{d}\right)$ with $\mathbf{d}=\left[d_{1},d_{2},\ldots,d_{M}\right]^{T}$. Subsequently, by applying the chain rule, $\boldsymbol{\Phi}_{\mathbf{s}}$ is expressed as
\begin{equation}
\boldsymbol{\Phi}_{\mathbf{s}}=\mathbf{J}\left(\mathbf{d}\right)\boldsymbol{\Phi}\left(\mathbf{d}\right)\left[\mathbf{J}\left(\mathbf{d}\right)\right]^{T},\label{fai_s_D}
\end{equation}
where $\mathbf{J}\left(\mathbf{d}\right)\in\mathbb{R}^{3\times M}$ denotes the Jacobian matrix of $\mathbf{d}$ obtained at the true localization target position $\mathbf{s}=\left[x^{s},y^{s},z^{s}\right]^{T}$, which is given by
\begin{equation}
\begin{split}\mathbf{J}\left(\mathbf{d}\right) & =\frac{\partial\mathbf{d}^{T}}{\partial\mathbf{s}}\end{split}
=\left[\begin{array}{cccc}
\frac{x_{1}^{q}-x^{s}}{d_{1}} & \frac{x_{2}^{q}-x^{s}}{d_{2}} & \ldots & \frac{x_{M}^{q}-x^{s}}{d_{M}}\\
\frac{y_{1}^{q}-y^{s}}{d_{1}} & \frac{y_{2}^{q}-y^{s}}{d_{2}} & \ldots & \frac{y_{M}^{q}-y^{s}}{d_{M}}\\
\frac{z_{1}^{q}-z^{s}}{d_{1}} & \frac{z_{2}^{q}-z^{s}}{d_{2}} & \ldots & \frac{z_{M}^{q}-z^{s}}{d_{M}}
\end{array}\right].\label{eq:J_d}
\end{equation}
Moreover, for convenience, we define the vector  $\hat{\mathbf{d}}=\left[\hat{d}_{1},\hat{d}_{2},\ldots,\hat{d}_{M}\right]^{T}$ to represent the measured distances between the ST and all UAVs. Based on \eqref{eq:d_measure}, we conclude that $\hat{\mathbf{d}}\sim\mathcal{N}\left(\mathbf{d},\boldsymbol{\Lambda}^{2}\right)$, where $\boldsymbol{\Lambda}^{2}$ is a diagonal matrix constructed from the elements $\left\{ \sigma_{m}^{2}\right\} _{m=1}^{M}$, which is denoted as
\begin{equation}
\boldsymbol{\Lambda}^{2}=\frac{\rho\sigma_{0}^{2}}{G_{p}\beta_{0}}\textrm{diag}\left(\frac{\left(d_{1}\right)^{4}}{p_{1}},\frac{\left(d_{2}\right)^{4}}{p_{2}},\ldots,\frac{\left(d_{M}\right)^{4}}{p_{M}}\right).
\end{equation}
According to \cite{CRB}, $\boldsymbol{\Phi}\left(\mathbf{d}\right)$ is given by
\begin{equation}
\begin{split}\left[\boldsymbol{\Phi}\left(\mathbf{d}\right)\right]_{\vartheta,\varrho} & =\left[\frac{\partial\mathbf{d}}{\partial d_{\vartheta}}\right]^{T}\left[\boldsymbol{\Lambda}^{2}\right]^{-1}\left[\frac{\partial\mathbf{d}}{\partial d_{\varrho}}\right] \\
&+\frac{1}{2}\textrm{tr}\left[\left[\boldsymbol{\Lambda}^{2}\right]\frac{\partial\left[\boldsymbol{\Lambda}^{2}\right]}{\partial d_{\vartheta}}\left[\boldsymbol{\Lambda}^{2}\right]^{-1}\frac{\partial\left[\boldsymbol{\Lambda}^{2}\right]}{\partial d_{\varrho}}\right],\vartheta,\varrho\in\mathcal{M}.\end{split}
\label{eq:fai_d}
\end{equation}
By substituting $\mathbf{J}\left(\mathbf{d}\right)$ and $\boldsymbol{\Phi}\left(\mathbf{d}\right)$ into \eqref{fai_s_D}, $\boldsymbol{\Phi}_{\mathbf{s}}$ is written as
\begin{equation}
\boldsymbol{\Phi}_{\mathbf{s}}=\left[\begin{array}{ccc}
\Phi_{11} & \Phi_{12} & \Phi_{13}\\
\Phi_{12} & \Phi_{22} & \Phi_{23}\\
\Phi_{13} & \Phi_{23} & \Phi_{33}
\end{array}\right],
\end{equation}
where $\Phi_{11}$, $\Phi_{22}$, $\Phi_{33}$, $\Phi_{12}$, $\Phi_{13}$, and $\Phi_{23}$ are given by ~\eqref{eq:fai_a}--\eqref{eq:fai_bc}, as shown at the top of the next page.
\begin{figure*}[t]
\begin{equation}
\Phi_{11}=\underset{m\in\mathcal{M}}{\sum}\left(\frac{p_{m}G_{p}\beta_{0}}{\rho\sigma_{0}^{2}}\frac{\left(x_{m}^{q}-x^{s}\right)^{2}}{\left(d_{m}\right)^{6}}+\frac{8\left(x_{m}^{q}-x^{s}\right)^{2}}{\left(d_{m}\right)^{4}}\right),\label{eq:fai_a}
\end{equation}

\begin{equation}
\Phi_{22}=\underset{m\in\mathcal{M}}{\sum}\left(\frac{p_{m}G_{p}\beta_{0}}{\rho\sigma_{0}^{2}}\frac{\left(y_{m}^{q}-y^{s}\right)^{2}}{\left(d_{m}\right)^{6}}+\frac{8\left(y_{m}^{q}-y^{s}\right)^{2}}{\left(d_{m}\right)^{4}}\right),\label{eq:fai_b}
\end{equation}
\begin{equation}
\Phi_{33}=\underset{m\in\mathcal{M}}{\sum}\left(\frac{p_{m}G_{p}\beta_{0}}{\rho\sigma_{0}^{2}}\frac{\left(z_{m}^{q}-z^{s}\right)^{2}}{\left(d_{m}\right)^{6}}+\frac{8\left(z_{m}^{q}-z^{s}\right)^{2}}{\left(d_{m}\right)^{4}}\right),\label{eq:fai_c}
\end{equation}
\begin{equation}
\Phi_{12}=\underset{m\in\mathcal{M}}{\sum}\left(\frac{p_{m}G_{p}\beta_{0}}{\rho\sigma_{0}^{2}}\frac{\left(x_{m}^{q}-x^{s}\right)\left(y_{m}^{q}-y^{s}\right)}{\left(d_{m}\right)^{6}}+\frac{8\left(x_{m}^{q}-x^{s}\right)\left(y_{m}^{q}-y^{s}\right)}{\left(d_{m}\right)^{4}}\right),\label{eq:fai_ab}
\end{equation}
\begin{equation}
\Phi_{13}=\underset{m\in\mathcal{M}}{\sum}\left(\frac{p_{m}G_{p}\beta_{0}}{\rho\sigma_{0}^{2}}\frac{\left(x_{m}^{q}-x^{s}\right)\left(z_{m}^{q}-z^{s}\right)}{\left(d_{m}\right)^{6}}+\frac{8\left(x_{m}^{q}-x^{s}\right)\left(z_{m}^{q}-z^{s}\right)}{\left(d_{m}\right)^{4}}\right),\label{eq:fai_ac}
\end{equation}
\begin{equation}
\Phi_{23}=\underset{m\in\mathcal{M}}{\sum}\left(\frac{p_{m}G_{p}\beta_{0}}{\rho\sigma_{0}^{2}}\frac{\left(y_{m}^{q}-y^{s}\right)\left(z_{m}^{q}-z^{s}\right)}{\left(d_{m}\right)^{6}}+\frac{8\left(y_{m}^{q}-y^{s}\right)\left(z_{m}^{q}-z^{s}\right)}{\left(d_{m}\right)^{4}}\right).\label{eq:fai_bc}
\end{equation}
\rule[0.5ex]{1\textwidth}{0.4pt}
\end{figure*}

Consequently, the joint optimization of the LQR cost, power allocation, UAV-robot association, and UAV positions is formulated as follows
\begin{subequations}\label{P1}
    \begin{align}
\min_{\mathbf{b},\boldsymbol{\Theta},\mathbf{p},\left\{ \mathbf{q}_{m}\right\} } & \quad\frac{\eta}{\Psi^c}\underset{k\in\mathcal{K}}{\sum}b_{k}-\frac{\left(1-\eta\right)}{\Psi^s}\det\boldsymbol{\Phi}_{\mathbf{s}}\label{P1-1}\\
\mathrm{\textrm{s.t.}} & \quad\underset{m\in\mathcal{M}}{\sum}p_{m}\leq P_{\max},\label{P1-C1}\\
 & \quad p_{m}\geq0,\forall m\in\mathcal{M},\label{P1-C2}\\
 & \quad0\leq\underset{k\in\mathcal{K}}{\sum}\theta_{m,k}\leq1,\forall m\in\mathcal{M},\label{P1-C3}\\
 & \quad\underset{m\in\mathcal{M}}{\sum}\theta_{m,k}=1,\forall k\in\mathcal{K},\label{P1-C4}\\
 & \quad\theta_{m,k}\in\left\{ 0,1\right\} ,\forall m\in\mathcal{M},k\in\mathcal{K},\label{P1-C5}\\
 & \quad\parallel\mathbf{q}_{m}-\mathbf{q}_{r}\parallel^{2}\geq d_{\min}^{2},\forall m,r\in\mathcal{M},m\neq r,\label{P1-C6}\\
 & \quad\mathbf{q}_{m}\in\mathcal{D},\forall m\in\mathcal{M},\label{P1-C7}\\
 & \quad B\underset{m\in\mathcal{M}}{\sum}\theta_{m,k}R_{m,k}\geq L_{k},\forall k\in\mathcal{K},\label{P1-C8}
 \end{align}
\end{subequations}
where $\ensuremath{\mathbf{b}=\left[b_{1},\ldots,b_{K}\right]^{T}}$ represents the LQR cost vector of the control system. In problem \eqref{P1}, $\eta\in\left[0,1\right]$ is a weighting factor, with a larger value of $\eta$ indicating a higher priority assigned to control performance compared to localization accuracy. To ensure a fair comparison between the control and sensing performance, we normalize these two performance increments via $\Psi^c$ and $\Psi^s$, which represent the upper bound of $\sum_{k\in\mathcal{K}}b_{k}$ and $\det\boldsymbol{\Phi}_{\mathbf{s}}$, respectively. Moreover, \eqref{P1-C1} denotes the power budget constraint, where $P_{\max}$ is the maximum transmission power of the network. Constraints \eqref{P1-C3}-\eqref{P1-C5} ensure that the association between UAVs and robots adheres to the criteria specified in Section \ref{subsec:Control-Model}. Furthermore, \eqref{P1-C6} and \eqref{P1-C7} correspond to the collision avoidance and flight boundary constraints for UAVs, where $d_{\min}$ denotes the minimum allowable distance between any two UAVs and $\mathcal{D}$ is the permissible flight area for each UAV, respectively.

The formulated problem \eqref{P1} is non-convex and involves mixed binary and continuous variables, making it difficult to achieve a globally optimal solution. To tackle this challenge, we propose an efficient iterative algorithm to solve problem \eqref{P1}, which will be detailed in the next section.

\section{Proposed Solutions \label{sec:Proposed-solutions}}
In this section, an AO-based algorithm is proposed to derive a sub-optimal solution to problem \eqref{P1}. Specifically, we employ a two-step method to initially optimize the LQR cost  $\mathbf{b}$ with respect to $\boldsymbol{\Theta}$, $\mathbf{P}$, and $\left\{ \mathbf{q}_{m}\right\} $. Then, problem \eqref{P1} can be transformed into an equivalent form to optimize the UAV-robot association, power allocation, and UAV positions, which is further decomposed into three sub-problems and solved using the AO method.

To start with, the following Proposition gives the closed-form of the optimal LQR cost $\mathbf{b}$, denoted as $\mathbf{b}^{\star}$, with respect to $\boldsymbol{\Theta}$, $\mathbf{P}$, and $\left\{ \mathbf{q}_{m}\right\} $. 

\begin{prop} \label{pro1}
Given arbitrary $\boldsymbol{\Theta}$, $\mathbf{P}$, and $\left\{ \mathbf{q}_{m}\right\} $, the closed-form expression of the optimal LQR cost to problem \eqref{P1} is given by
\begin{equation}
b_{k}^{\star}=b_{k}\left(\boldsymbol{\Theta},\mathbf{p},\left\{ \mathbf{q}_{m}\right\}\right)=\frac{\iota\left(\det\mathbf{N}_{k}\mathbf{M}_{k}\right)^{\frac{1}{\iota}}}{2^{f_{k}\left(\boldsymbol{\Theta},\mathbf{p},\left\{\mathbf{q}_{m}\right\}\right)}-1}+\left(b_{k}\right)_{\min},\forall k\in\mathcal{K}\label{rct},
\end{equation}
with
\begin{equation}
f_{k}\left(\mathbf{p},\boldsymbol{\Theta},\left\{ \mathbf{q}_{m}\right\}\right)=\frac{2}{\iota}\left(B\underset{m\in\mathcal{M}}{\sum}\theta_{m,k}R_{m,k}-g_{k}\right),\forall k\in\mathcal{K}.
\end{equation}
\end{prop}

\begin{proof}
Note that only constraint \eqref{P1-C8} imposes restrictions on $\mathbf{b}$ in problem \eqref{P1}. Hence, given arbitrary $\boldsymbol{\Theta}$, $\mathbf{P}$, and $\left\{ \mathbf{q}_{m}\right\} $, optimizing $\mathbf{b}$ reduces to solving the following problem
\begin{subequations}\label{P11}
    \begin{align}
\min_{\mathbf{b}} & \quad\frac{\eta}{\Psi^c}\underset{k\in\mathcal{K}}{\sum}b_{k}\label{P11-1}\\
\mathrm{\textrm{s.t.}} & \quad B\underset{m\in\mathcal{M}}{\sum}\theta_{m,k}R_{m,k}\geq L_{k},\forall k\in\mathcal{K}.\label{P11-C1}
\end{align}
\end{subequations}
It is observed that the RHS of \eqref{P11-C1} is a monotonically decreasing function of $b_{k}$. Thus, the feasible solution of $b_{k}$ is achieved by taking equality in \eqref{P11-C1} to minimize $\sum_{k\in\mathcal{K}}b_{k}$, which is equivalent to \eqref{rct}, completing the proof. \footnote{In this work, we consider an assuredly stable scenario, where the data throughput is substantially larger than the intrinsic entropy rate, i.e., $B\sum_{m\in\mathcal{M}}\theta_{m,k}R_{m,k}\gg g_{k},\forall k\in\mathcal{K}$, to guarantee a positive denominator in \eqref{rct} \cite{cont&comm4}. This implies that the control system operates well away from the instability threshold, allowing the focus to shift towards optimizing control performance rather than ensuring system stability.}
\end{proof}

Using Proposition \ref{pro1}, we derive the closed-form expression of the optimal LQR cost to problem \eqref{P1} with respect to $\boldsymbol{\Theta}$, $\mathbf{P}$, and $\left\{ \mathbf{q}_{m}\right\} $. Therefore, plugging \eqref{rct} into \eqref{P1}, the optimization problem is reformulated as
\begin{subequations} 
\label{P2}
\begin{align}
\min_{\boldsymbol{\Theta},\mathbf{p},\left\{ \mathbf{q}_{m}\right\} } & 
\frac{\eta}{\Psi^c}\underset{k\in\mathcal{K}}{\sum}b_{k}\left(\boldsymbol{\Theta},\mathbf{p},\left\{ \mathbf{q}_{m}\right\} \right)-\frac{\left(1-\eta\right)}{\Psi^s}\det\boldsymbol{\Phi}_{\mathbf{s}}\label{P2-1}\\
\mathrm{\textrm{s.t.}} & \quad\underset{m\in\mathcal{M}}{\sum}p_{m}\leq P_{\max},\label{P2-C1}\\
 & \quad p_{m}\geq0,\forall m\in\mathcal{M},\label{P2-C2}\\
 & \quad0\leq\underset{k\in\mathcal{K}}{\sum}\theta_{m,k}\leq1,\forall m\in\mathcal{M},\label{P2-C3}\\
 & \quad\underset{m\in\mathcal{M}}{\sum}\theta_{m,k}=1,\forall k\in\mathcal{K},\label{P2-C4}\\
 & \quad\theta_{m,k}\in\left\{ 0,1\right\} ,\forall m\in\mathcal{M},k\in\mathcal{K},\label{P2-C5}\\
 & \quad\parallel\mathbf{q}_{m}-\mathbf{q}_{r}\parallel^{2}\geq d_{\min}^{2},\forall m,r\in\mathcal{M},m\neq r,\label{P2-C6}\\
 & \quad\mathbf{q}_{m}\in\mathcal{D},\forall m\in\mathcal{M}.\label{P2-C7}
\end{align}
\end{subequations}
Although problem \eqref{P2} is easier to handle, it is still challenging to solve due to its complex objective function and the coupling of mixed binary and continuous variables. We then decouple problem \eqref{P2} into three sub-problems, i.e., UAV-robot association, power allocation, and UAV position design.

\subsection{UAV-Robot Association}

In this subsection, we focus on optimizing the UAV-robot association with the given power allocation and UAV positions. The corresponding optimization problem is formulated as
\begin{subequations}\label{P3}
    \begin{align}
\min_{\boldsymbol{\Theta}} & \quad\varphi\left(\boldsymbol{\Theta}\right)\label{P3-1}\\
\mathrm{\textrm{s.t.}} & \quad0\leq\underset{k\in\mathcal{K}}{\sum}\theta_{m,k}\leq1,\forall m\in\mathcal{M},\label{P3-C1}\\
 & \underset{m\in\mathcal{M}}{\sum}\theta_{m,k}=1,\forall k\in\mathcal{K},\label{P3-C2}\\
 & \quad\theta_{m,k}\in\left\{ 0,1\right\} ,\forall m\in\mathcal{M},k\in\mathcal{K},\label{P3-C3}
\end{align}
\end{subequations}
where
\begin{equation}
\varphi\left(\boldsymbol{\Theta}\right)=
\frac{\eta}{\Psi^c}\underset{k\in\mathcal{K}}{\sum}b_{k}\left(\boldsymbol{\Theta}\right)-\frac{\left(1-\eta\right)}{\Psi^s}\det\boldsymbol{\Phi}_{\mathbf{s}}.
\end{equation}

To solve the UAV-robot association problem, the concave-convex property of the function $b_{k}\left(\boldsymbol{\Theta}\right)$ is analyzed first. For convenience, the expression is rewritten as
\begin{equation}
b_{k}\left(\boldsymbol{\Theta}\right)=\frac{\Omega_{k}}{2^{f_{k}\left(\boldsymbol{\Theta}\right)}-1}+\left(b_{k}\right)_{\min},
\end{equation}
with $\Omega_{k}=\iota\left(\det\mathbf{N}_{k}\mathbf{M}_{k}\right)^{\frac{1}{\iota}}$.

\begin{prop}
\label{thm:Function--is}Function $b_{k}\left(\boldsymbol{\Theta}\right)$
is convex with respect to $\boldsymbol{\Theta}$ under the assuredly stable scenario.
\end{prop}
\begin{proof}
The first order derivation of function $b_{k}\left(\boldsymbol{\Theta}\right)$ is given by
\begin{equation}
b_{k}^{\prime}\left(\boldsymbol{\Theta}\right)=-\Omega_{k}\ln2\frac{2^{f_{k}\left(\boldsymbol{\Theta}\right)}f_{k}^{\prime}\left(\boldsymbol{\Theta}\right)}{\left(2^{f_{k}\left(\boldsymbol{\Theta}\right)}-1\right)^{2}},
\end{equation}
where $f_{k}^{\prime}\left(\boldsymbol{\Theta}\right)=\frac{2}{\iota}B\sum_{m\in\mathcal{M}}R_{m,k}>0.$ Subsequently, the second order derivation of function $b_{k}\left(\boldsymbol{\Theta}\right)$ is calculated as
\begin{equation}
b_{k}^{\prime\prime}\left(\boldsymbol{\Theta}\right)=2\Omega_{k}\left(\ln2\right)^{2}\left[f_{k}^{\prime}\left(\boldsymbol{\Theta}\right)\right]^{2}f_{k}\left(\boldsymbol{\Theta}\right)\frac{2^{f_{k}\left(\boldsymbol{\Theta}\right)}+1}{\left(2^{f_{k}\left(\boldsymbol{\Theta}\right)}-1\right)^{3}}.\label{sod}
\end{equation}
It follows that $b_{k}^{\prime\prime}\left(\boldsymbol{\Theta}\right)>0$ holds as long as $2^{f_{k}\left(\boldsymbol{\Theta}\right)}>1$, which is equivalent to $B\sum_{m\in\mathcal{M}}\theta_{m,k}R_{m,k}\gg g_{k}$. This condition is always satisfied under the assuredly stable scheme. Therefore, the function $b_{k}\left(\boldsymbol{\Theta}\right)$ is convex under the assuredly stable scenario, completing the proof.
\end{proof}

Proposition \ref{thm:Function--is} indicates the convexity of the objective function of problem \eqref{P3}. However, it is still challenging to tackle directly due to the constraint \eqref{P3-C3}. By converting
$\theta_{m,k}$ into an equivalent continuous form, problem \eqref{P3} is equivalently transformed into
\begin{subequations}\label{P4}
    \begin{align}
\min_{\boldsymbol{\Theta}} & \quad\frac{\eta}{\Psi^c}\underset{k\in\mathcal{K}}{\sum}b_{k}\left(\boldsymbol{\Theta}\right)\label{P4-1}\\
\mathrm{s.t.} & \quad0\leq\theta_{m,k}\leq1,\forall m\in\mathcal{M},k\in\mathcal{K},\label{P4-C1}\\
 & \quad\underset{m\in\mathcal{M}}{\sum}\underset{k\in\mathcal{K}}{\sum}\left(\theta_{m,k}-\theta_{m,k}^{2}\right)\leq0,\label{P4-C2}\\
 & \quad\eqref{P3-C1},\eqref{P3-C2},\label{P4-C3}
\end{align}
\end{subequations}
where constraints \eqref{P4-C1} and \eqref{P4-C2} ensure that the value of $\theta_{m,k}$ is restricted to either $0$ or $1$, which is equivalent to constraint \eqref{P3-C3}. Although \eqref{P4-C2} remains non-convex, it has been expressed as the difference of two convex functions, making it solvable via DC programming \cite{DC}. To approximate \eqref{P4-C2}, we employ the first-order Taylor series expansion of $\sum_{m\in\mathcal{M}}\sum_{k\in\mathcal{K}}\theta_{m,k}^{2}$
as
\begin{equation}
\underset{m\in\mathcal{M}}{\sum}\underset{k\in\mathcal{K}}{\sum}\left(\theta_{m,k}^{\left(t_{1}\right)}\right)^{2}+2\underset{m\in\mathcal{M}}{\sum}\underset{k\in\mathcal{K}}{\sum}\theta_{m,k}^{\left(t_{1}\right)}\left(\theta_{m,k}-\theta_{m,k}^{\left(t_{1}\right)}\right)\leq0,\label{SCA-theta}
\end{equation}
where $t_{1}$ denotes the $t_{1}$-th iteration of the DC programming. 
\textcolor{black}{Since the left-hand-side (LHS) of \eqref{SCA-theta} cannot be smaller than $0$, finding a solution directly is challenging. To circumvent this difficulty, we then incorporate \eqref{SCA-theta} into the objective function using the penalty technique, further transforming problem \eqref{P4} into\begin{subequations}\label{P5}
    \begin{align}
\min_{\boldsymbol{\Theta}} & \quad\frac{\eta}{\Psi^c}\underset{k\in\mathcal{K}}{\sum}b_{k}\left(\boldsymbol{\Theta}\right)+\mu\underset{m\in\mathcal{M}}{\sum}\underset{k\in\mathcal{K}}{\sum}\theta_{m,k}+\phi\left(\boldsymbol{\Theta}\right),\label{P5-1}\\
\mathrm{s.t.} & \quad\eqref{P3-C1},\eqref{P3-C2},\eqref{P4-C1},\label{P5-C1}
\end{align}
\end{subequations}
where 
\begin{equation}
\phi\left(\boldsymbol{\Theta}\right)=\mu\underset{m\in\mathcal{M}}{\sum}\underset{k\in\mathcal{K}}{\sum}\left(\left(\theta_{m,k}^{\left(t_{1}\right)}\right)^{2}+2\theta_{m,k}^{\left(t_{1}\right)}\left(\theta_{m,k}-\theta_{m,k}^{\left(t_{1}\right)}\right)\right),
\end{equation}
and $\mu$ is a non-negative parameter used to penalize the objective function. }

The detailed steps of the penalty-DC algorithm are outlined in Algorithm \ref{alg:1}. First, the feasible starting point $\boldsymbol{\Theta}^{\left(t_{1}\right)}$, the initial penalty parameter $\mu^{\left(t_{1}\right)}$, and $\omega^{\left(t_{1}\right)}$, which represents the value of the objective function of problem \eqref{P5}, are initialized. Next, problem \eqref{P5} is solved via convex optimization tools, such as CVX \cite{CVX}. Moreover, it is worth noting that the penalty parameter $\mu$ can be initialized with a small value and increased gradually by multiplying it with a positive constant $A$ until problem \eqref{P5} converges. i.e., $\mid\left(\omega^{\left(t_{1}\right)}-\omega^{\left(t_{1}-1\right)}\right)/\omega^{\left(t_{1}-1\right)}\mid<\varepsilon$, with $\varepsilon$ denoting the maximum permissible error.
\begin{algorithm}[t]
\caption{Penalty-DC Optimization Algorithm for UAV-Robot Association\label{alg:1}}

1: \textbf{Initialize} $t_{1}\leftarrow0$, $\boldsymbol{\Theta}^{\left(t_{1}\right)}$
, $\mu^{\left(t_{1}\right)}$, $\omega^{\left(t_{1}\right)}$, maximum number  

$\quad\qquad\qquad$ of iterations $t_{1,\max}$, tolerance $\varepsilon$, and initial  

$\quad\qquad\qquad$ penalty parameter $\mu_{\max}$.

2: \textbf{repeat}

3: $\qquad$$t_{1}\leftarrow t_{1}+1$.

4: $\qquad$Solve problem \eqref{P5} to obtain $\boldsymbol{\Theta}^{\left(t_{1}\right)}$and
$\omega^{\left(t_{1}\right)}$.

5: $\qquad$Update $\mu^{\left(t_{1}\right)}\leftarrow\min\left\{ A\mu^{\left(t_{1}-1\right)},\mu_{\max}\right\} $.

6: \textbf{until} $\mid\frac{\omega^{\left(t_{1}\right)}-\omega^{\left(t_{1}-1\right)}}{\omega^{\left(t_{1}-1\right)}}\mid<\varepsilon$
or $t_{1}>t_{1,\max}$.

7: \textbf{Output} $\boldsymbol{\Theta}^{\star}\leftarrow\boldsymbol{\Theta}^{\left(t_{1}\right)}$.
\end{algorithm}

\subsection{Power Allocation}

In this section, we optimize the power allocation under fixed UAV positions and the updated UAV-robot association. The corresponding power allocation problem is given by
\begin{subequations}\label{P7}
    \begin{align}
\min_{\mathbf{p}} & \quad\varphi\left(\mathbf{p}\right)\label{P7-1}\\
\mathrm{s.t.} & \quad\underset{m\in\mathcal{M}}{\sum}p_{m}\leq P_{\max},\label{P7-C1}\\
 & \quad p_{m}\geq0,\forall m\in\mathcal{M},\label{P7-C2}
\end{align}
\end{subequations}
where
\begin{equation}
\varphi\left(\mathbf{p}\right)=\frac{\eta}{\Psi^c}\underset{k\in\mathcal{K}}{\sum}b_{k}\left(\mathbf{p}\right)-\frac{\left(1-\eta\right)}{\Psi^s}\det\boldsymbol{\Phi}_{\mathbf{s}}\left(\mathbf{p}\right).
\end{equation}
Note that problem \eqref{P7} is a linearly constrained optimization problem with a non-convex objective function. Thus, it can be optimally solved through the PGD method \cite{PGD}. The pseudocode of the PGD-based algorithm is given in Algorithm \ref{alg:2-1}. In Algorithm \ref{alg:2-1}, the process begins with initializing a feasible starting point, assuming equal power allocation among all UAVs, along with an initial step size $\rho$ and a convergence tolerance $\varepsilon$. At each iteration, the search direction is determined by computing the gradient of $\varphi\left(\mathbf{p}\right)$ with respect to $\mathbf{p}$. Then, $\mathbf{p}$ is updated through a projection onto the feasible set $\mathcal{P}$, i.e., $\Psi_{\mathcal{P}}\left(\mathbf{p}^{\prime}-\rho\triangle\right):\mathbb{R}^{M}\rightarrow\mathcal{P}$, which is obtained by solving the following convex problem
\begin{subequations} \label{P8}
    \begin{align}
\min_{\mathbf{p}} & \quad\parallel\mathbf{p}^{\prime}-\rho\triangle-\mathbf{p}\parallel\label{P8-1}\\
\mathrm{s.t.} & \quad\eqref{P7-C1},\eqref{P7-C2}.\label{P8-C1}
\end{align}
\end{subequations}
Furthermore, the step size $\rho$ is gradually reduced by multiplying it with $\frac{1}{1+\hat{\rho}}$, until $\parallel\mathbf{p}-\mathbf{p}^{\prime}\parallel<\varepsilon$, where $\hat{\rho}$ is a small positive constant.

\begin{algorithm}[t]
\caption{PGD-Based Algorithm for Power Allocation\label{alg:2-1}}

1: \textbf{Initialize} A starting feasible point $\mathbf{p}\leftarrow\left[\frac{P_{\max}}{M},\ldots,\frac{P_{\max}}{M}\right]$, 

$\quad\qquad\qquad$ initial step size $\rho$, and tolerance $\varepsilon$.

2: \textbf{repeat}

3: $\qquad$Save the previous direction vector $\mathbf{p}^{\prime}\leftarrow\mathbf{p}$.

4: $\qquad$Determine a search direction $\triangle\triangleq\nabla\varphi_{\mathbf{p}}$.

5: $\qquad$Choose a step size $\rho\leftarrow\frac{\rho}{\parallel\nabla\varphi_{\mathbf{p}}\parallel}$.

6: $\qquad$Update and $\mathbf{p}\leftarrow\Psi_{\mathcal{P}}\left(\mathbf{p}^{\prime}-\rho\triangle\right)$.

7: \textbf{until} $\parallel\mathbf{p}-\mathbf{p}^{\prime}\parallel<\varepsilon$

8: \textbf{Output} $\mathbf{p}^{\star}\leftarrow\mathbf{p}$.
\end{algorithm}

\subsection{UAV Positions Optimization and Overall Algorithm}

In this section, we focus on optimizing the UAV positions based on the optimized UAV-robot association and power allocation. The corresponding optimization problem is formulated as
\begin{subequations}\label{P9}
    \begin{align}
\min_{\left\{ \mathbf{q}_{m}\right\} } & \quad\varphi\left(\left\{ \mathbf{q}_{m}\right\} \right)\label{P9-1}\\
\mathrm{\textrm{s.t.}} & \quad\parallel\mathbf{q}_{m}-\mathbf{q}_{r}\parallel^{2}\geq d_{\min}^{2},\forall m,r\in\mathcal{M},m\neq r,\label{P9-C1}\\
 & \quad\mathbf{q}_{m}\in\mathcal{D},\forall m\in\mathcal{M},\label{P9-C2}
\end{align}
\end{subequations}
with
\begin{equation}
\varphi\left(\left\{ \mathbf{q}_{m}\right\} \right)=\frac{\eta}{\Psi^c}\underset{k\in\mathcal{K}}{\sum}b_{k}\left(\left\{ \mathbf{q}_{m}\right\} \right)-\frac{\left(1-\eta\right)}{\Psi^s}\det\boldsymbol{\Phi}_{\mathbf{s}}\left(\left\{ \mathbf{q}_{m}\right\}\right).
\end{equation}
Problem \eqref{P9} is difficult to solve due to the non-convex objective function and constraint \eqref{P9-C1}. Therefore, the successive convex approximation (SCA) methodology is employed to approximate the non-convex terms via first-order Taylor series expansion \cite{SCA_CC}. Since constraint \eqref{P9-C1} is convex with respect to both $\mathbf{q}_{m}$ and $\mathbf{q}_{r}$, it is lower bounded by \cite{dmin_convex}
\begin{equation}
\begin{split}\parallel\mathbf{q}_{m}-\mathbf{q}_{r}\parallel^{2}\geq & 2\left(\mathbf{q}_{m}^{\left(t_{2}\right)}-\mathbf{q}_{r}^{\left(t_{2}\right)}\right)^{T}\left(\mathbf{q}_{m}-\mathbf{q}_{r}\right)\\
 & -\parallel\mathbf{q}_{m}^{\left(t_{2}\right)}-\mathbf{q}_{r}^{\left(t_{2}\right)}\parallel^{2},\forall m\neq r,
\end{split}
\label{SCA-dmin}
\end{equation}
with $\mathbf{q}_{m}^{\left(t_{2}\right)}$ and $\mathbf{q}_{r}^{\left(t_{2}\right)}$ denoting the approximation of the $m$-th and $j$-th UAV's position at iteration $t_{2}$. Next, we address the non-convex objective function. It is worth noting that the function $\det\boldsymbol{\Phi}\left(\mathbf{s}\right)$ involves the 3D coordinates $x_{m}^{q}$, $y_{m}^{q}$, and $z_{m}^{q}$, which are coupled with each other, thereby complicating the optimization of UAV positions significantly. To simplify the optimization, we reformulate ~\eqref{eq:fai_a}--\eqref{eq:fai_bc} with respect to $\mathbf{q}_{m}$ instead of the individual coordinates $x_{m}^{q}$, $y_{m}^{q}$, and $z_{m}^{q}$. The reformulated expressions are presented in \eqref{eq:fai}, as shown at the top of the next page, with $\boldsymbol{\xi}_{1}=\left[1,0,0\right]^{T}$, $\boldsymbol{\xi}_{2}=\left[0,1,0\right]^{T}$, and $\boldsymbol{\xi}_{3}=\left[0,0,1\right]^{T}$, respectively.
\begin{figure*}[tp]
\begin{equation}
\Phi_{ij}=\underset{m\in\mathcal{M}}{\sum}\left(\frac{P_{m}G_{p}\beta_{0}}{\rho\sigma_{0}^{2}}\frac{\left(\boldsymbol{\xi}_{i}^{T}\left[\mathbf{q}_{m}-\mathbf{s}\right]\right)\left(\boldsymbol{\xi}_{j}^{T}\left[\mathbf{q}_{m}-\mathbf{s}\right]\right)}{\left(d_{m}\right)^{6}}+\frac{8\left(\boldsymbol{\xi}_{i}^{T}\left[\mathbf{q}_{m}-\mathbf{s}\right]\right)\left(\boldsymbol{\xi}_{j}^{T}\left[\mathbf{q}_{m}-\mathbf{s}\right]\right)}{\left(d_{m}\right)^{4}}\right),i,j=1,2,3.\label{eq:fai}
\end{equation}

\rule[0.5ex]{1\textwidth}{0.4pt}
\end{figure*}
Subsequently, $\varphi\left(\left\{ \mathbf{q}_{m}\right\} \right)$ is approximated as \eqref{SCA-detJ} at the given feasible point $\mathbf{q}_{m}^{\left(t_{2}\right)}$ using the first-order Taylor series expansion.
\begin{equation}
    \psi=\varphi\left(\left\{ \mathbf{q}_{m}^{\left(t_{2}\right)}\right\} \right)+\stackrel[m=1]{M}{\sum}\nabla\varphi_{\mathbf{q}_{m}}\left[ \mathbf{q}_{1}^{\left(t_{2}\right)},...,\mathbf{q}_{M}^{\left(t_{2}\right)}\right]^{T}\left(\mathbf{q}_{m}-\mathbf{q}_{m}^{\left(t_{2}\right)}\right).\label{SCA-detJ}
\end{equation}
As a result, the problem \eqref{P9} is approximate to solving a set of convex optimization problems
\begin{subequations}\label{P12}
    \begin{align} 
\min_{\left\{ \mathbf{q}_{m}\right\} } & \quad\psi\label{P12-1}\\
\mathrm{\textrm{s.t.}} & \quad2\left(\mathbf{q}_{m}^{\left(t_{2}\right)}-\mathbf{q}_{r}^{\left(t_{2}\right)}\right)^{T}\left(\mathbf{q}_{m}-\mathbf{q}_{r}\right)\nonumber \\
 & \quad-\parallel\mathbf{q}_{m}^{\left(t_{2}\right)}-\mathbf{q}_{r}^{\left(t_{2}\right)}\parallel^{2}\geq d_{\min}^{2},\forall m,r\in\mathcal{M},m\neq r,\label{P12-C1}\\
 & \quad\eqref{P9-C2},\label{P12-C2}
\end{align}
\end{subequations}
which can be solved using standard convex optimization tools such as CVX \cite{CVX}.

\begin{algorithm}[t]
\caption{AO-Based Algorithm for Solving Original Problem \eqref{P1} \label{alg:3}}

1: \textbf{Initialize} $t\leftarrow0$, $\mathbf{p}^{\left(t\right)}$,
$\left\{ \mathbf{q}_{m}^{\left(t\right)}\right\} $, $\delta^{\left(t\right)}$,
tolerance $\varepsilon$, and

$\quad\qquad\qquad$  maximum number of iterations $t_{\max}$,

2: \textbf{repeat}

3: $\qquad$$t\leftarrow t+1$.

4: $\qquad$Solve problem \eqref{P3} using Algorithm \ref{alg:1}
with given

$\quad\qquad$ $\mathbf{p}^{\left(t-1\right)}$ and $\left\{ \mathbf{q}_{m}^{\left(t-1\right)}\right\} $
to obtain $\boldsymbol{\Theta}^{\left(t\right)}$.

5: $\qquad$Solve problem \eqref{P7} using Algorithm \ref{alg:2-1},
$\boldsymbol{\Theta}^{\left(t\right)}$, and 

$\quad\qquad$ $\left\{ \mathbf{q}_{m}^{\left(t-1\right)}\right\} $
to obtain $\mathbf{p}^{\left(t\right)}.$

6: $\qquad$Solve problem \eqref{P9} using SCA, $\boldsymbol{\Theta}^{\left(t\right)}$, and $\mathbf{p}^{\left(t\right)}$ to  

$\quad\qquad$  obtain $\left\{ \mathbf{q}_{m}^{\left(t\right)}\right\} $.

7: \textbf{until} $\mid\frac{\delta^{\left(t\right)}-\delta^{\left(t-1\right)}}{\delta^{\left(t-1\right)}}\mid<\varepsilon$
or $t>t_{\max}$.

8: \textbf{Output }$\boldsymbol{\Theta}^{\star}\leftarrow\boldsymbol{\Theta}^{\left(t\right)}$,
$\mathbf{p}^{\star}\leftarrow\mathbf{p}^{\left(t\right)}$, and
 $\left\{ \mathbf{q}_{m}^{\star}\right\} \leftarrow\left\{ \mathbf{q}_{m}^{\left(t\right)}\right\} $.
\end{algorithm}

Finally, an AO-based algorithm is proposed to alternately optimize the UAV-robot association, power allocation, and UAV positions. The detailed steps are summarized in Algorithm \ref{alg:3}. In Algorithm \ref{alg:3}, we first initialize the $\mathbf{p}^{\left(t\right)}$ and $\left\{ \mathbf{q}_{m}^{\left(t\right)}\right\} $ that satisfy constraints \eqref{P2-C1}, \eqref{P2-C2}, \eqref{P2-C6}, and \eqref{P2-C7}, the maximum number of iterations $t_{\max}$, and the value of the objective function $\delta^{\left(t\right)}$ of problem \eqref{P2}. Next, problems \eqref{P3}, \eqref{P7}, and \eqref{P9} are solved alternately with other optimized variables until convergence. Since the optimal solutions of problems \eqref{P3}, \eqref{P7}, and \eqref{P9} only provide a near-optimal solution to problem \eqref{P2}, analyzing the convergence of the AO-based algorithm is essential. 
\textcolor{black}{For simplicity, we define $\delta\left(\mathbf{p},\boldsymbol{\Theta},\left\{ \mathbf{q}_{m}\right\} \right)$, $\delta_{1}\left(\mathbf{p},\boldsymbol{\Theta},\left\{ \mathbf{q}_{m}\right\} \right)$, $\delta_{2}\left(\mathbf{p},\boldsymbol{\Theta},\left\{ \mathbf{q}_{m}\right\} \right)$, and $\delta_{3}\left(\mathbf{p},\boldsymbol{\Theta},\left\{ \mathbf{q}_{m}\right\} \right)$ as the objective function of problems \eqref{P2}, \eqref{P3}, \eqref{P7}, and \eqref{P9}, respectively. In the fourth step of Algorithm \ref{alg:3}, the optimal $\boldsymbol{\Theta}^{\left(t\right)}$ is obtained by the given $\mathbf{p}^{\left(t-1\right)}$ and $\left\{ \mathbf{q}_{m}^{\left(t-1\right)}\right\} $. Thus, we have
\begin{equation}
\begin{split}\delta\left(\mathbf{p}^{\left(t-1\right)},\boldsymbol{\Theta}^{\left(t-1\right)},\left\{ \mathbf{q}_{m}^{\left(t-1\right)}\right\} \right) & =\delta_{1}\left(\mathbf{p}^{\left(t-1\right)},\boldsymbol{\Theta}^{\left(t-1\right)},\left\{ \mathbf{q}_{m}^{\left(t-1\right)}\right\} \right)\\
 & \geq\delta_{1}\left(\mathbf{p}^{\left(t-1\right)},\boldsymbol{\Theta}^{\left(t\right)},\left\{ \mathbf{q}_{m}^{\left(t-1\right)}\right\} \right).
\end{split}
\label{C1}
\end{equation}
Next, in the fifth step, given $\left\{ \mathbf{q}_{m}^{\left(t-1\right)}\right\} $ and the updated $\boldsymbol{\Theta}^{\left(t\right)}$, it holds that
\begin{equation}
\begin{split}\delta_{1}\left(\mathbf{p}^{\left(t-1\right)},\boldsymbol{\Theta}^{\left(t\right)},\left\{ \mathbf{q}_{m}^{\left(t-1\right)}\right\} \right) & \geq\delta_{2}\left(\mathbf{p}^{\left(t\right)},\boldsymbol{\Theta}^{\left(t\right)},\left\{ \mathbf{q}_{m}^{\left(t-1\right)}\right\} \right).\end{split}
\label{C2}
\end{equation}
Similarly, following the sixth step, we obtain the $\left\{ \mathbf{q}_{m}^{\left(t\right)}\right\} $ with optimized $\boldsymbol{\Theta}^{\left(t\right)}$and $\mathbf{p}^{\left(t\right)}$, which implies
\begin{equation}
\begin{split}
\delta_{2}\left(\mathbf{p}^{\left(t\right)},\boldsymbol{\Theta}^{\left(t\right)},\left\{ \mathbf{q}_{m}^{\left(t-1\right)}\right\} \right) & \geq\delta_{3}\left(\mathbf{p}^{\left(t\right)},\boldsymbol{\Theta}^{\left(t\right)},\left\{ \mathbf{q}_{m}^{\left(t\right)}\right\} \right)\\
 & =\delta\left(\mathbf{p}^{\left(t\right)},\boldsymbol{\Theta}^{\left(t\right)},\left\{ \mathbf{q}_{m}^{\left(t\right)}\right\} \right),\label{C3}
\end{split}
\end{equation}
Consequently, it follows that
\begin{equation}
\begin{split}\delta\left(\mathbf{p}^{\left(t-1\right)},\boldsymbol{\Theta}^{\left(t-1\right)},\left\{ \mathbf{q}_{m}^{\left(t-1\right)}\right\} \right) & \geq\delta\left(\mathbf{p}^{\left(t\right)},\boldsymbol{\Theta}^{\left(t\right)},\left\{ \mathbf{q}_{m}^{\left(t\right)}\right\} \right)\end{split}
,\label{C2-1}
\end{equation}
which indicates that the objective value of problem \eqref{P2} decreases with each iteration. Furthermore, considering the objective function of \eqref{P2} is lower bounded by a finite value due to the minimum LQR cost $\left(b_{k}\right)_{\min}$ and the upper bound $\Psi^s$ on sensing accuracy, the convergence of the proposed AO-based algorithm is guaranteed.}

Next, we analyze the computational complexity of Algorithm \ref{alg:3}. The computational complexity of the algorithm \ref{alg:1} for UAV-robot association primarily arises from problem \eqref{P5} at each iteration, which is solved via the standard interior-point method. Given the iteration number $T_{1}$, the complexity of Algorithm \ref{alg:1} is $C_{1}=\mathcal{O}\left(T_{1}\left(MK\right)^{3.5}\right)$ \cite{SCA_CC}. The computational complexity of algorithm \ref{alg:2-1} lies in the gradient calculation of $\varphi\left(\mathbf{p}\right)$ with respect to $\mathbf{p}$ at each iteration, resulting in the computational complexity of Algorithm \ref{alg:2-1} being $C_{2}=\mathcal{O}\left(\log\left(1/\varepsilon^{2}\right)KM^{2}\right)$ \cite{PGD}. The UAV positions are optimized using the standard interior-point method and SCA approach, whose computational complexity is $C_{3}=\mathcal{O}\left(T_{3}\left(3M\right)^{3.5}\right)$ with $T_{3}$ representing the iteration number of SCA. Consequently, the total computational complexity of Algorithm \ref{alg:3} is $C_{tot}=\mathcal{O}\left(T\left(C_{1}+C_{2}+C_{3}\right)\right)$, with $T$ denoting the required iteration numbers.

\section{Simulation Results\label{sec:4}}
\begin{figure}
\centering
\includegraphics[width=0.9\linewidth]{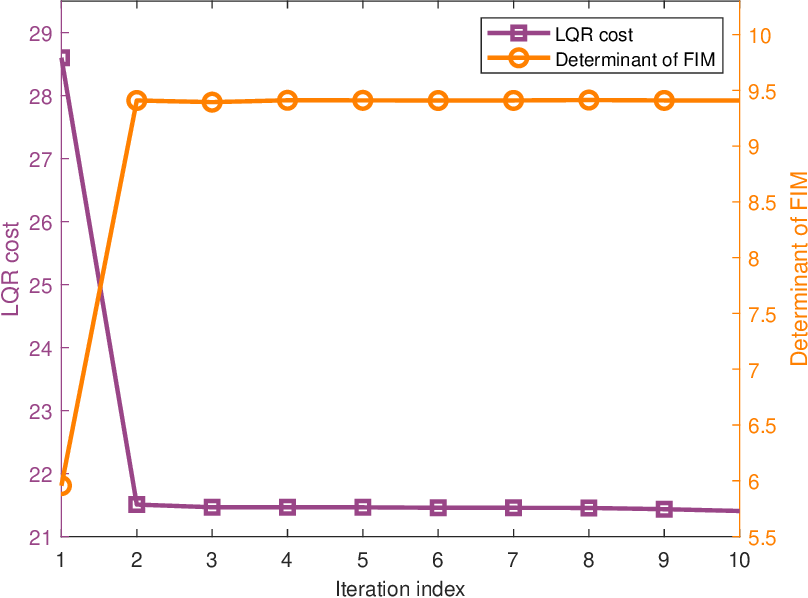}
\caption{\label{fig:Convergence}Convergence of the proposed algorithm,
with $\eta=0.5$, $P_{\textrm{max}}=-1\ \textrm{dBW}$, $\boldsymbol{\Sigma}_{k}^{w}=0.001\times\mathbf{I}_{\zeta}$, and $l_{m,k}=1024\ \textrm{\textrm{bit}}$.}
\end{figure}
In this section, we present the simulation results derived from the proposed algorithm, which simultaneously optimizes control and sensing performance in a multi-UAV cooperative network with FBL transmission. In the simulation, we assume that $4$ UAVs fly in a permissible flight area $\mathcal{D}$, defined as a rectangular region with dimensions $100\times100\ \textrm{m}^{2}$ and a fixed height $z_{m}^{q}=100\ \textrm{m},\forall m\in\mathcal{M}$. The minimum collision avoidance distance is set to $d_{\min}=25\ \textrm{m}$. Furthermore, we set $\alpha_{0}=-49\ $dB and $\beta_{0}=-50\ $dB, respectively, and the noise power at the receiver is set to $\sigma_{k}^{2}=\sigma_{0}^{2}=-110\ \textrm{dBm},\forall k \in\mathcal{K}$. The bandwidth of the shared frequency channel is set as $B=500\ \textrm{kHz}$, while other parameters are set as $G_{p}=0.1\times B$ and $\rho=200$ \cite{UAV_S&C3,UAV_S&C4}. Regarding control parameters, the intrinsic entropy rate $g_{k}$ is randomly selected from $g_{k}\in\left[0,50\right],\forall k\in\mathcal{K}$. The dimensions of the system state and observation result are both set to $25$, i.e., $\iota=\zeta=25$. The covariance matrix of the system noise is defined as $\boldsymbol{\Sigma}_{k}^{v}=0.001\times\mathbf{I}_{\iota}$. The control input matrix and observation matrix are both identity matrices, and the weight matrices of LQR cost are set as $\mathbf{Q}_{k}=\mathbf{I}_{\iota}$ and $\mathbf{R}_{k}=\mathbf{0}$ \cite{cont&comm4}. 
\textcolor{black}{Moreover, the normalization parameters, $\Psi^c$ and $\Psi^s$, are set to $30$ and $12$, respectively.} 
Throughout the simulations, control performance, sensing performance, and UAV positions are analyzed under variations in the power budget, observation noise variance, blocklength of the command signals, and weighting factor. 

\textcolor{black}{Before evaluating the system performance, we introduce several benchmark schemes for comparison. \emph{(a)} \textbf{Equal power allocation}: This scheme distributes the available power budget $P_{\max}$ evenly among all $M$ UAVs, without considering the specific control or sensing requirements of individual UAVs. \textit{(b)} \textbf{Random UAV positioning}: In this scheme, UAV deployment positions are randomly generated while ensuring compliance with collision avoidance constraints and flight boundaries. \textit{(c)} \textbf{Water-filling method}: This approach dynamically allocates power to UAVs based on their channel conditions, prioritizing UAVs with stronger communication links while maintaining an adaptive power distribution strategy \cite{water_filling}. \textit{(d)} \textbf{Sensing-only scheme}: This scheme focuses solely on optimizing sensing performance without considering control objectives.}

\begin{figure*}
    \centering
    % 第一张子图
    \begin{subfigure}{0.3\linewidth}
        \centering
        \includegraphics[width=\linewidth]{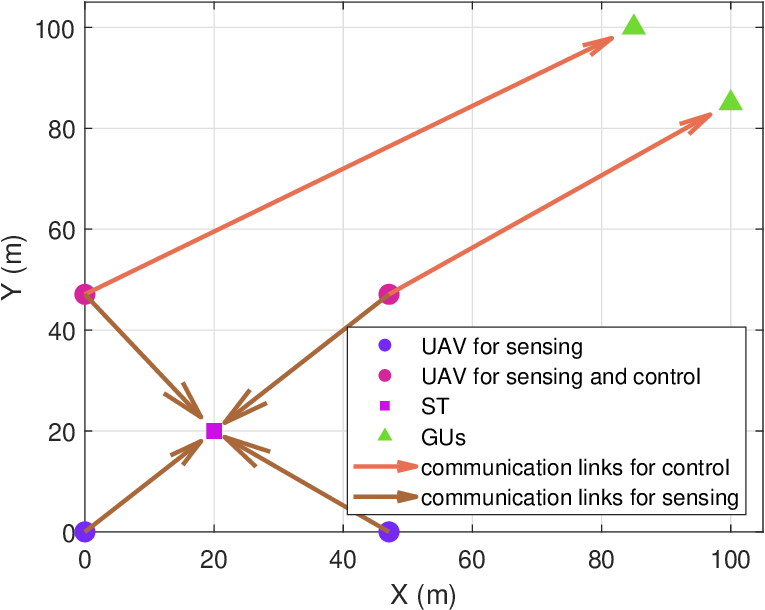}
        \caption{ \label{UPa}} % 仅显示 (a)
    \end{subfigure}
    \hfill
    \begin{subfigure}{0.3\linewidth}
        \centering
        \includegraphics[width=\linewidth]{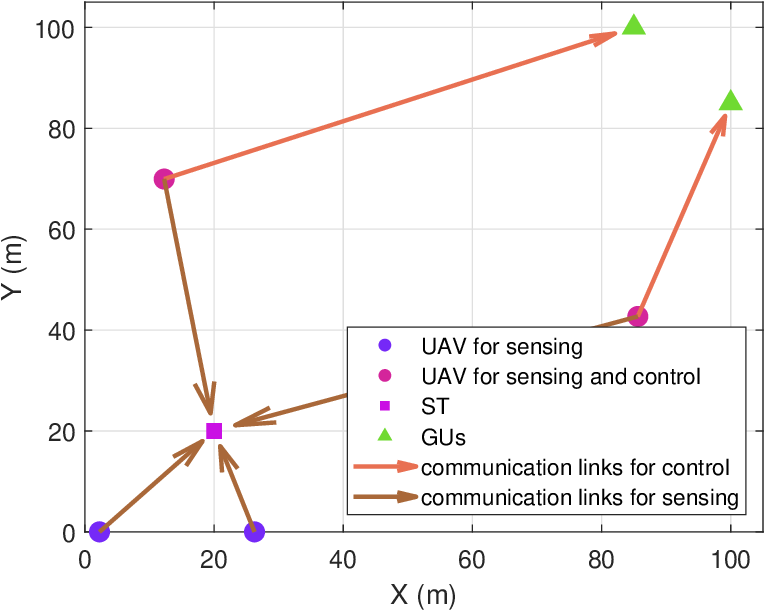}
        \caption{\label{UPb}} % 仅显示 (a)
    \end{subfigure}
    \hfill
    \begin{subfigure}{0.3\linewidth}
        \centering
        \includegraphics[width=\linewidth]{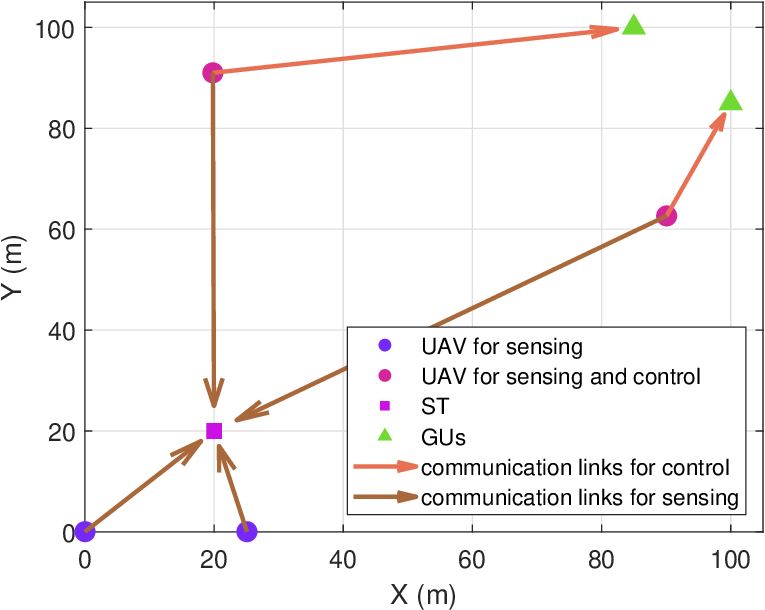}
        \caption{\label{UPc}} % 仅显示 (a)
    \end{subfigure}
\caption{\label{fig:UAV-positions}UAV positions and UAV-robot association with
different weighting factors with $P_{\textrm{max}}=-1\ \textrm{dBW}$, $\boldsymbol{\Sigma}_{k}^{w}=0.001\times\mathbf{I}_{\zeta}$, and $l_{m,k}=1024\ \textrm{\textrm{bit}}$. (a) $\eta=0.1$. (b) $\eta=0.5$. (c) $\eta=0.9$.}

\end{figure*}

\begin{figure}
\centering
\includegraphics[width=0.9\linewidth]{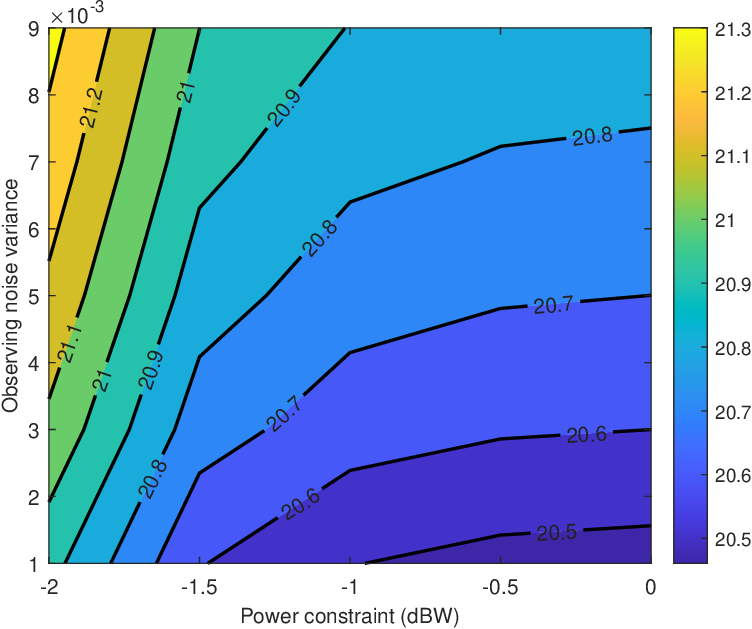}
\caption{\label{fig:LQR}LQR cost with different power constraints
and observation noise variances with $\eta=0.5$ and $l_{m,k}=1024\ \textrm{bit}$.}
\end{figure}

\begin{figure}
\centering
\includegraphics[width=0.9\linewidth]{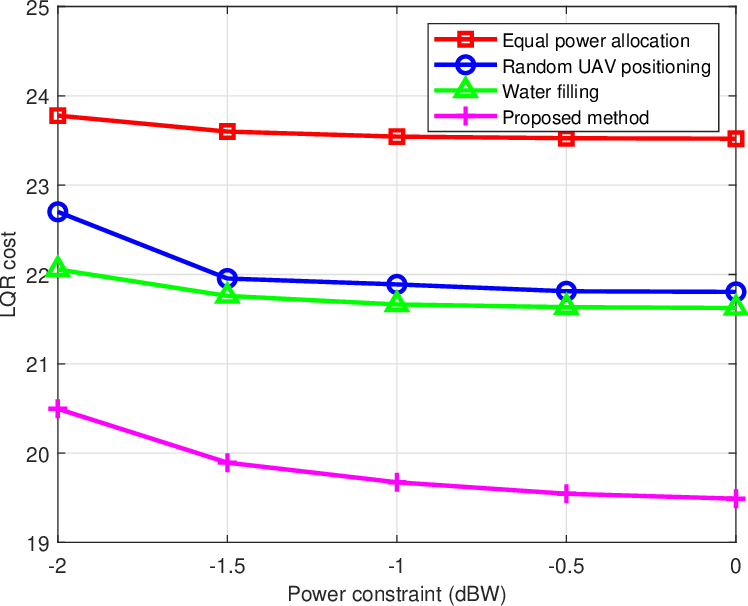}
\caption{\label{fig:LQRbm}LQR cost versus power constraint under different schemes with $\eta=0.5$ and $\boldsymbol{\Sigma}_{k}^{w}=0.001\times\mathbf{I}_{\zeta}$.}
\end{figure}

The convergence of the proposed algorithm is first analyzed. As illustrated in Fig. \ref{fig:Convergence}, both the LQR cost and the determinant of the FIM exhibit rapid monotonic changes, achieving optimal solutions within only a few iterations. This result indicates the convergence of the overall algorithm.

\textcolor{black}{Fig. \ref{fig:UAV-positions} illustrates the impact of the weighting factor $\eta$ on UAV positions and their association with robots. When $\eta=0.1$ (Fig. \ref{UPa}), UAVs are primarily positioned around the ST, focusing on sensing tasks, as indicated by the strong presence of sensing communication links. As $\eta$ increases to 0.5 (Fig. \ref{UPb}), UAVs gradually shift their positions toward the robots, balancing sensing and control tasks. When $\eta$ further increases to 0.9 (Fig. \ref{UPc}), UAVs are predominantly positioned near the robots, prioritizing control over sensing. However, UAVs dedicated to sensing remain close to the ST regardless of $\eta$, demonstrating that their role remains independent of the weighting factor. This shift in UAV positions highlights the trade-off between sensing and control performance, with UAVs adapting their locations based on the relative importance of these objectives.}

% Before revision: Fig. \ref{fig:UAV-positions} illustrates the impact of the weighting factor $\eta$ on UAV positions and their association with robots. When $\eta$ is small, as shown in Fig. \ref{UPa}, UAVs are primarily positioned around the ST, focusing on sensing tasks. As $\eta$ increases, from Fig.  \ref{UPb} to Fig. \ref{UPc}, the UAVs gradually shift their positions toward the robots, prioritizing control tasks over sensing. This reflects the increasing influence of control performance on the overall objective function. However, UAVs that are dedicated solely to sensing remain close to the sensing target, indicating their role is independent of $\eta$. The observed shift in UAV positions highlights the underlying trade-off between sensing and control performance. 

Next, we explore the impact of observation and communication capabilities on control system performance. Fig. \ref{fig:LQR} depicts the contours of the LQR cost with varying power budget $P_{\textrm{max}}$ and observation noise variances $\sigma_{w}^{2}$ ($\boldsymbol{\Sigma}_{k}^{w}=\sigma_{w,k}^{2}\times\mathbf{I}_{\zeta}$ and $\sigma_{w,k}^{2}=\sigma_{w}^{2},\forall k\in\mathcal{K}$). As seen from the picture, the LQR cost increases with higher observation noise variance, suggesting that control performance deteriorates with reduced observation capability. This degradation can be mitigated by enhancing communication capabilities, as the LQR cost decreases with an increasing power budget $P_{\textrm{max}}$. Moreover, when the power budget becomes sufficiently large, the LQR cost decreases slowly and converges to the lower bound $\left(b_{k}\right)_{\min}$, which can be further reduced by enhancing the observation capability. This result highlights the nonlinear interplay between control performance and communication capability, emphasizing that enhancing communication alone cannot fully counteract the adverse effects of diminished observation capability.

Based on Fig. \ref{fig:LQRbm}, the proposed algorithm demonstrates significant advantages over the other methods in terms of minimizing the LQR cost, which directly indicates better control performance. Compared to the equal power allocation and random UAV positioning schemes, our proposed approach achieves consistently lower LQR costs across all power constraint levels. Specifically, while the equal power allocation method shows a nearly constant and higher LQR cost due to the lack of optimization of power allocation, and the random UAV positioning method exhibits marginal improvements, the proposed algorithm effectively leverages power optimization and UAV positioning to significantly reduce the LQR cost. Additionally, compared to the water-filling method, which also performs better than the baseline schemes, the proposed method achieves superior results, particularly under stricter power constraints. This highlights the efficiency and robustness of the proposed method for optimizing both power allocation and UAV positioning for enhanced control performance.

Next, to examine the impact of FBL transmission on control performance, Fig. \ref{fig:FBL} illustrates the relationship between the LQR cost and the blocklength of the transmitted signals under various power budgets, i.e., $P_{\textrm{max}}=-3\ \textrm{dBW}$, $P_{\textrm{max}}=-2\ \textrm{dBW}$, $P_{\textrm{max}}=-1\ \textrm{dBW}$, and $P_{\textrm{max}}=-0\ \textrm{dBW}$. Note that a longer blocklength leads to higher data throughput, which, in turn, enhances communication capabilities and improves the control performance. Thus, we observe that the LQR cost decreases as the blocklength increases. Similar to Fig. \ref{fig:LQR}, one can find that the LQR cost remains constrained by a lower bound, even when the blocklength becomes sufficiently large. This limitation arises because, although a larger blocklength enhances communication capability, the resulting improvement in control performance is fundamentally limited and eventually converges toward a scenario resembling infinite blocklength transmission. 

\begin{figure}
\centering
\includegraphics[width=0.9\linewidth]{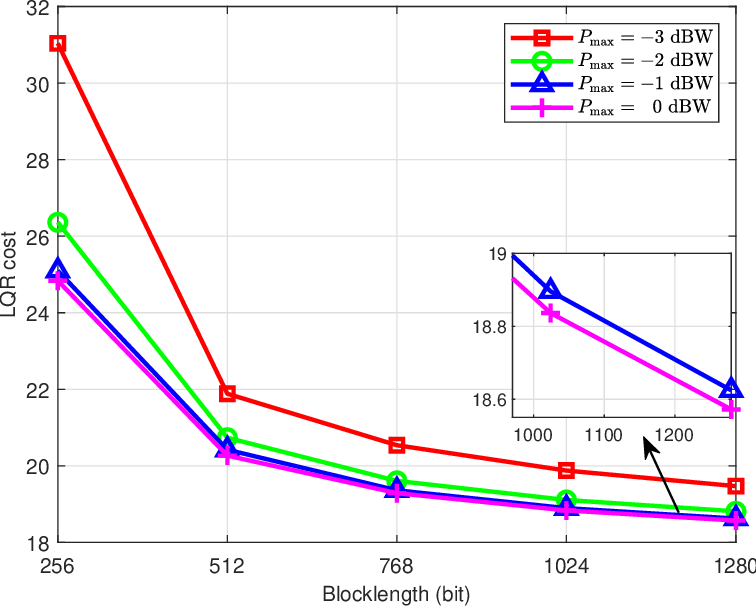}
\caption{\label{fig:FBL}LQR cost versus the blocklength of the ISAC signals with $\eta=0.5$ and $\boldsymbol{\Sigma}_{k}^{w}=0.001\times\mathbf{I}_{\zeta}$.}
\end{figure}

\textcolor{black}{Fig. \ref{fig:sensing} depicts the sensing performance versus power constraint. Rather than using the determinant of the FIM, we employ the sum CRB of $x^{s}$, $y^{s}$, and $z^{s}$ as the estimation metric, as it directly provides a lower bound for the root mean squared error (RMSE), making it a more intuitive measure of localization accuracy. The CRB is obtained by computing the trace of the inverse of the FIM, following standard estimation theory \cite{CRB}. The RMSE of the ST estimation is evaluated through a Monte Carlo simulation with $100$ trials. In addition, to highlight the advantages of our joint optimization algorithm, we compare the proposed algorithm with the sensing-only scheme, where the power budget is set to half of that used in the co-design scheme to ensure fairness. As expected, it is evident that the joint optimization scheme outperforms the sensing-only scheme in terms of sensing performance. Moreover, as shown in Fig. \ref{fig:sensing}, both CRB and RMSE decrease with an increase in power budget, which is because the rise in power contributes to better SNR of the echoes received by the UAV, enabling a greater localization performance for the ST. However, despite the improvement in sensing performance with an increasing power budget, the CRB and RMSE decrease slowly when the power budget becomes large. This occurs because once SNR reaches a sufficiently high level, the marginal impact of additional power increases diminishes, leading to a slower reduction in CRB and RMSE.}
%Before revision: Fig. \ref{fig:sensing} depicts the sensing performance versus power constraint.  Rather than using the determinant of the FIM, we employ the sum CRB of $x^{s}$, $y^{s}$, and $z^{s}$ as the estimation metric, which provides a lower bound for the root mean squared error (RMSE) and can be obtained by computing the trace of the inverse matrix of FIM \cite{CRB} The RMSE of the ST estimation is evaluated through a Monte Carlo simulation with $100$ trials. In addition, to highlight the advantages of our joint optimization algorithm, we compare the proposed algorithm with the sensing-only scheme, where the power budget is set to half of that used in the co-design scheme to ensure fairness. As expected, it is evident that the joint optimization scheme outperforms the sensing-only scheme in terms of sensing performance. Moreover, as shown in Fig. \ref{fig:sensing}, both CRB and RMSE decrease with an increase in power budget, which is because the rise in power contributes to better SNR of the echoes received by the UAV, enabling a greater localization performance for the ST. However, despite the improvement in sensing performance with an increasing power budget, the CRB and RMSE decrease slowly when the power budget becomes large. This occurs because once SNR reaches a sufficiently high level, the marginal impact of additional power increases diminishes, leading to a slower reduction in CRB and RMSE.. 

\begin{figure}
\centering
\includegraphics[width=0.9\linewidth]{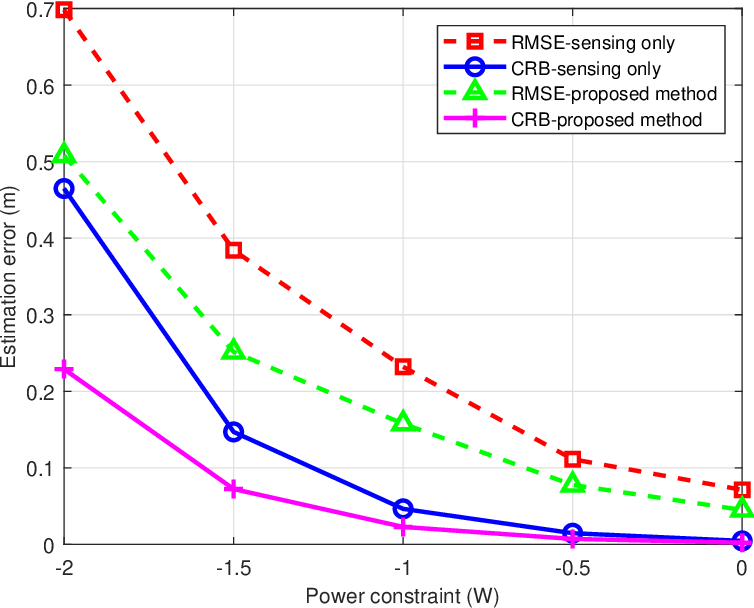}
\caption{\label{fig:sensing}Sensing performance versus power constraint with $\eta=0.5$, $\boldsymbol{\Sigma}_{k}^{w}=0.001\times\mathbf{I}_{\zeta}$, and $l_{m,k}=1024\ \textrm{\textrm{bit}}$.}
\end{figure}

To better understand the trade-off between sensing and control performance, Fig. \ref{fig:weight-factor} examines how the weighting factor $\eta$ quantitatively affects the performance metrics, i.e., the LQR cost for control and the CRB for sensing. It can be observed that there exists a clear trade-off between the LQR cost and the CRB. Specifically, as $\eta$ increases, the LQR cost gradually decreases, indicating improved control performance, while the CRB increases notably, reflecting a degradation in sensing accuracy. Moreover, Fig. \ref{fig:weight-factor} reveals that the CRB is more sensitive to changes in $\eta$, particularly at larger $\eta$ values, compared to the LQR cost, highlighting the asymmetry in their responses to the weighting factor. This demonstrates that the trade-off between control and sensing performance can be adjusted, but the degree of influence differs for the two metrics. 

\begin{figure}
    \centering
    % 第一张子图
    \begin{subfigure}{0.9\linewidth}
        \centering
        \includegraphics[width=\linewidth]{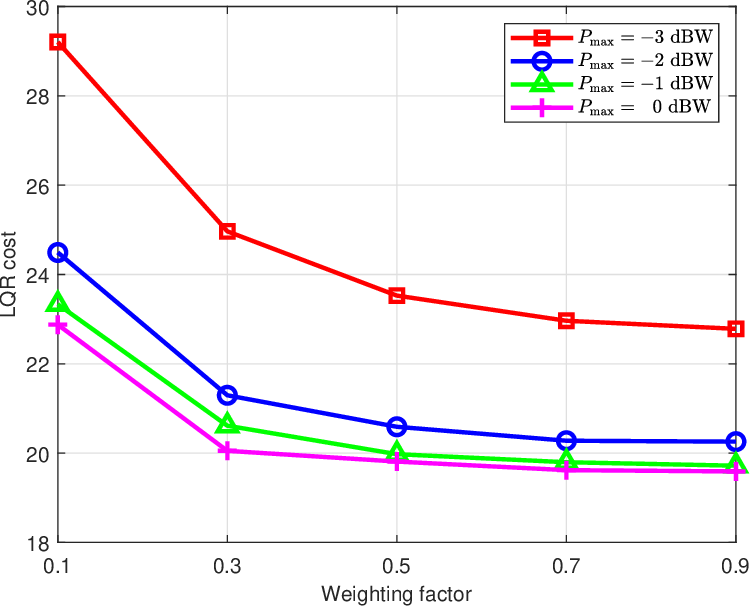}
        \caption*{(a)} % 仅显示 (a)
    \end{subfigure}
    % 第二张子图
    \begin{subfigure}{0.9\linewidth}
        \centering
        \includegraphics[width=\linewidth]{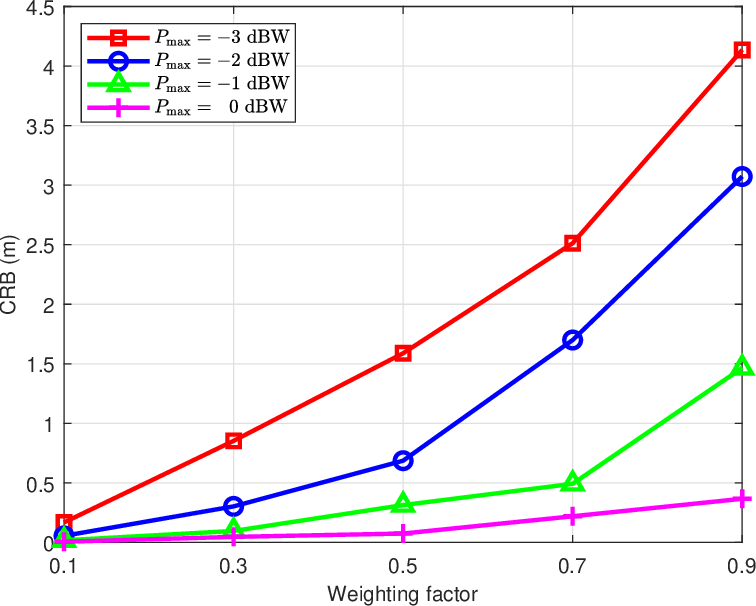}
        \caption*{(b)} % 仅显示 (b)
    \end{subfigure}
    \caption{The trade-off between control and sensing versus weighting factor $\eta$ with $\boldsymbol{\Sigma}_{k}^{w}=0.001\times\mathbf{I}_{\zeta}$ and $l_{m,k}=1024\ \textrm{\textrm{bit}}$. (a) LQR cost versus weighting factor $\eta$. (b) CRB versus weighting factor $\eta$.}
    \label{fig:weight-factor}
\end{figure}

\section{Conclusion\label{sec:5}}

This work explored the co-design of integrated sensing, communication, and control in a multi-UAV cooperative system with FBL transmission. The LQR cost, power allocation, UAV-robot association, and UAV positions were jointly optimized to maximize sensing and control performance while considering rate-LQR cost bounds, power budget, collision avoidance, and flight boundary constraints. To tackle the problem, we first derived a closed-form expression of the optimal LQR cost concerning other variables. The problem was then decomposed into three sub-problems, which were solved alternately using the AO method. Specifically, the UAV-robot association was optimized via DC programming, with the convexity of the objective function analyzed. Next, the PGD method was employed to optimize power allocation. Subsequently, UAV positions were designed based on the updated LQR cost, power allocation, and UAV-robot association. Finally, an efficient AO-based algorithm was proposed, and its convergence properties and computational complexity were thoroughly analyzed. Simulation results were presented to demonstrate the effectiveness of the proposed optimization algorithm.

\bibliographystyle{IEEEtran}
\bibliography{a_plot_reference/reference}

\end{document}